\documentclass{article}
\usepackage[margin=1in]{geometry}
\usepackage{proof}
\usepackage{booktabs}
\usepackage{lineno,hyperref}
\usepackage{scalerel}
\usepackage{amsthm}
\usepackage{accsupp} 
\usepackage{soul}

\usepackage{etex}
\usepackage{microtype}
\usepackage{centernot}
\usepackage{stmaryrd}
\usepackage{wrapfig}
\usepackage{amsfonts}
\usepackage{mathtools}
\usepackage{xspace}
\usepackage{color}
\usepackage{todonotes}
\usepackage{spot}
\usepackage{esvect}
\definecolor{lgray}{gray}{0.9}
\usepackage{fixltx2e}
\usepackage{amsmath}
\usepackage{amssymb}
\usepackage{relsize}
\usepackage{fancybox}
\usepackage[all]{xy}\CompileMatrices
\usepackage{stackengine}
\stackMath
\usepackage{ifthen}
\usepackage{pgf}
\usepackage{tikz}
\usetikzlibrary{arrows,automata}
\usepackage{proof}

\usepackage{calrsfs}
\usepackage{scalerel}
\newcommand{\rulename}[1]{\text{\textsc{#1}}}

\DeclareMathAlphabet{\pazocal}{OMS}{zplmf}{m}{n}
\newcommand{\mcal}[1]{\pazocal{#1}}

\newcommand{\rTo}[1]{\xrightarrow{#1}}

\newcommand{\true}{{\sf tt}}

\newlength{\arrow}
\newcounter{sqindex}

\usepackage{fancybox}
\newcommand{\abacus}{{\rm{\it A\kern-.13emb}\kern-.15em\raise1.5ex\hbox{{\it a}}\kern-.3em{\it C}uS}\xspace}
\newcommand{\goat}{{\rm{\it $\mcal{G}$\kern-.13em o$\mcal{A}$t}}\xspace}
\usepackage{subfig}
 \newcommand{\rom}[1]{ \textup{(\lowercase\expandafter{\romannumeral#1})}}
\usepackage[ruled]{algorithm2e} 

\SetAlFnt{\small}
\SetAlCapFnt{\small}
\SetAlCapNameFnt{\small}
\SetAlCapHSkip{0pt}
\IncMargin{-\parindent}

\newcommand{\mset}[1]{\{{#1}\}}
\newcommand{\mconfig}[1]{\langle{#1}\rangle}

\newcommand{\qt}[1]{``{#1}"}

\newcommand{\sys}[2]{{#1};{#2}}
\newcommand{\fun}{f}
\newcommand{\func}[2]{\fun^{#1}_{#2}}
\newcommand{\gfunc}[3]{{#1}^{#2}_{#3}} 

\newcommand{\station}[6]{{#1}[{#2}, {#3}, {#4}, {#5}, {#6}]}

\newcommand{\zero}{{\bf 0}}
\newcommand{\varX}{X}
\newcommand{\rec}[1]{{{\bf rec}\ \varX.}{#1}}

\newcommand{\filter}[2]{[\func{d}{}]^{#1}_{#2}}
\newcommand{\bfilter}[2]{[\func{\toall}{}]^{#1}_{#2}}
\newcommand{\sfilter}[2]{[\func{\toself}{}]^{#1}_{#2}}
\newcommand{\ufilter}[2]{[\func{\toup}{}]^{#1}_{#2}}

\newcommand{\gbfilter}[3]{[\gfunc{#1}{\toall}{}]^{#2}_{#3}} 
\newcommand{\gsfilter}[3]{[\gfunc{#1}{\toself}{}]^{#2}_{#3}}
\newcommand{\gufilter}[3]{[\gfunc{#1}{\toup}{}]^{#2}_{#3}}
\newcommand{\grfilter}[3]{[\gfunc{#1}{\toright}{}]^{#2}_{#3}}

\newcommand{\toall}{\star}
\newcommand{\toself}{\bullet} 
\newcommand{\toup}{\blacktriangle}
\newcommand{\toright}{\blacktriangleright}

\newcommand{\parop}{\,|\,}
\newcommand{\plusop}{+}
\newcommand{\id}{\mathit{id}}
\newcommand{\ids}{I}
\newcommand{\idop}[1]{(#1)}
\newcommand{\upd}{\mathtt{upd}}

\newcommand{\node}[2]{#1:#2}
\newcommand{\nparop}{\, \|\,}
\newcommand{\state}{s}

\newcommand{\source}[2]{{#1},{#2}}
\newcommand{\sourcestate}{t}
\newcommand{\nei}{n}
\newcommand{\capa}{k}
\newcommand{\act}{a}
\newcommand{\err}{e}

\newcommand{\red}{\xrightarrow{}}

\newcommand{\proj}[3]{[\![ #1 ]\!]^{#2}_{#3}}

\newcommand{\tproj}[1]{[\![ #1 ]\!]}

\newcommand{\deff}{\triangleq}

\newcommand{\pid}{\mathtt{r}}

\newcommand{\net}{{\cal N}}

\newcommand{\link}{\mcal{L}}
\newcommand{\linktr}[2]{{#1}^{#2}}

\newcommand{\gfilter}[4]{[\gfunc{#1}{#2}{}]^{#3}_{#4}}
\newtheorem{example}{Example}[section]
\newtheorem{definition}{Definition}[section]
\newtheorem{lemma}{Lemma}[section]

\newtheorem{theorem}{Theorem}[section]
\newtheorem{proposition}{Proposition}[section]
\usepackage{todonotes}
\usepackage{authblk}
\begin{document}
\date{}
\title{Operation Control Protocols in Power Distribution Grids\thanks{Yehia Abd Alrahman has been partially funded by the ERC consolidator grant DSynMA under the European Union’s Horizon 2020 research and innovation programme (grant agreement No 772459).}}

\author[1,2]{Yehia Abd Alrahman\thanks{yehia.abdalrahman@leicester.ac.uk}}
\author[1]{Hugo Torres Vieira\thanks{hugotvieira@imtlucca.it}}
\affil[1]{IMT School for Advanced Studies, Lucca, Italy}
\affil[2]{University of Leicester, Leicester, UK}
\maketitle

\begin{abstract}
Future power distribution grids will comprise a large number of components, 
each potentially able to carry out operations autonomously. Clearly, in order to ensure 
safe operation of the grid, individual operations must be coordinated among 
the different components. Since operation safety is a global property, modelling component 
coordination typically involves reasoning about systems at a global level. In this paper,
we propose a language for specifying grid operation control protocols from a 
global point of view. We show how such global specifications can be used to automatically
generate local controllers of individual components, and that the distributed implementation
yielded by such controllers operationally corresponds to the global specification.
We showcase our development by modelling a fault management scenario
in power grids.
\end{abstract}


\section{Introduction}
\label{sec:introduction}
Modern power grids are large systems comprising a large number of geographically dispersed components. For example, the transmission lines of the North 
American power grid link all electricity generation and distribution in the continent~\cite{amin}. Because such grids are highly interconnected, any failure 
even in one location might be magnified as it propagates through the grid's infrastructure to cause instantaneous impacts on a wide area, causing what 
is called \emph{Cascading effects}~\cite{cascade}. This means that power grids might exhibit a global failure almost instantaneously as a result of a local 
one. It is therefore crucial that future power grids are equipped with mechanisms that
enable self-healing, and support automatic power restoration, voltage control, and power flow management. 

Current power grids suffer from the fact that 
control operations and power flow management are done by relying on a central control through the supervisory control and data acquisition (SCADA) 
system~\cite{scada}. Although the operational principles of interconnected power grids have been established ever since the 1960s, the coordination of 
control operations across the infrastructure of a power grid has received less attention~\cite{amin}. This is due to the fact that the operational principles 
were established before the emergence of powerful computers and communication networks. Computation is now largely used in all levels of the power 
grid, but coordination happens on a slower rate. Actually, only specific coordination operations are carried out by computer control while the rest is still 
based on telephone calls between system operators at the utility control centres even during emergencies. Furthermore, protection systems are limited 
to manage specific infrastructure components only.

Clearly, the current centralised control and protection systems cannot cope with the level of dynamicity and demands of modern power grids. Future power 
grids should be both flexible and extensible. Flexibility in the sense of their ability to reconfigure and adapt in response to failures. Extensible in the 
sense that new components can be added or removed from the grid without compromising its overall operation. 
Such type of grids would ease the integration of distributed generation (DG) (e.g., using renewable energy) and energy storage. A promising idea is to allow 
distributed control instead of the centralised one, as proposed by the European SmartGrids Technology Platform~\cite{euogrid} and the IntelliGrid 
project~\cite{intelli}.
%

The idea is to view the power grid as a collection of independent and autonomous substations, collaborating to achieve a desired global goal. It is then 
necessary to equip each substation with an independent controller that is able to communicate and cooperate with others, forming a large distributed 
computing system. Each controller must be connected to sensors associated with its own substation so that it can assess its own status and report them 
to its neighbouring controllers via communication paths. The communication paths of the controllers should follow the electrical connection paths. In this 
way, each substation can have a local view of the grid, represented by the connections to its neighbours. 

Despite the many advances of tools~\cite{aura,mattool,mas}, theoretical foundation to fully design and manage power systems is still lacking~\cite{formal}. 
In this paper, we propose a formal model of distributed power systems (Sect.~\ref{sec:definition}), and we show how automatic synthesis of substation 
controllers can be obtained (Sect.~\ref{sec:corress}). More precisely, we provide a compact high-level language that can be used to specify the operation 
control protocols governing the actual behaviour of the whole power grid from a global point of view. Also, we show how our global model can be used to 
automatically synthesise individual controllers that yield a distributed implementation. To showcase our development we model a fault management 
scenario
(Sect.~\ref{sec:scenario}).

The main advantage of our approach is that it reduces the engineering complexity of power grids by means of  decentralising decision making among 
substations. Thus, power grid management is carried out locally, avoiding overloading the control centre both at the level of the communication 
infrastructure and of operation control. Furthermore, the global model is more amenable to verify (global) system-level properties. Thanks to an 
operational correspondence result between the global model and the automatically generated distributed model, we can also ensure that any property 
that holds for the global specification also holds for the distributed implementation. 

In what follows, we informally present the language that supports
the specification of operation control protocols from a \emph{global}
perspective, i.e., considering the set of interactions as a whole. 
For the sake of a more intuitive reading, we introduce and motivate 
our design choices by pointing to the fault management scenario which will be fully explained in Sect.~\ref{sec:scenario}. In Sect.~\ref{sec:definition} and 
Sect.~\ref{sec:corress}, we present the syntax and the semantics of the 
global and the distributed language respectively. In Sect.~\ref{sec:corress}, we also show how to automatically synthesise the operationally precise behaviour of individual 
substations given the global specification of the protocol. In Sect.~\ref{sec:scenario}, we show how natural and intuitive it is to program a sizeable and 
interesting case study from the realm of power grids using the global syntax and we then show how the synthesis works. Finally, in 
Sect.~\ref{sec:related}, we comment on research directions and present related work.
\subsection{Operation Control Protocols, Informally}
\label{sec:preview}
We consider a scenario involving a power distribution grid where a 
fault occurs in one of the 
transmission power lines, and a recovery protocol that locates and
isolates the fault by means of interactions among the nodes in the grid. 

A central notion in our proposal is that nodes involved in such a protocol yield
control by means of synchronisation actions. For example, 
consider a node that is willing to $\mathsf{Locate}$ the fault and another
one that can react and continue the task of locating the fault.
The two nodes may synchronise on action $\mathsf{Locate}$
and the enabling node yields the control to the reacting one.
We therefore consider that \emph{active} nodes enable synchronisations
and, as a consequence of a synchronisation, transfer the \emph{active} role
 to the reacting node. So, for example, after some node $1$ synchronises 
 with some node $2$ on $\mathsf{Locate}$, node $2$ may enable the following interaction, 
 e.g., a synchronisation with some node $3$ also on $\mathsf{Locate}$.
 
We may therefore specify protocols as a (structured) set of interactions
without identifying the actual nodes involved a priori, since these are determined operationally
due to the transference of the active role in synchronisations. 
We write $\gfilter{\mathsf{Locate}}{\_}{\_}{\_}P$ to specify a protocol stating that a synchronisation on $\mathsf{Locate}$ is to take place first,
after which the protocol proceeds as specified by $P$ (for now we abstract
from the remaining elements using $\_$). 
Also, we write $\idop{1}P'$ to specify
that node $1$ is active on protocol $P'$. So, by 
$\idop{1}\gfilter{\mathsf{Locate}}{\_}{\_}{\_}P$
we represent that node $1$ may enable the synchronisation on $\mathsf{Locate}$.
Furthermore, we may write 
$\idop{1}\gfilter{\mathsf{Locate}}{\_}{\_}{\_}P \rTo{} 
\gfilter{\mathsf{Locate}}{\_}{\_}{\_}\idop{2}P$ to represent such a synchronisation step 
($\rTo{}$) between nodes $1$ and $2$, 
where node $1$ yields the control and node $2$ is activated to carry out
the continuation protocol $P$. This allows to capture the transference of the active role
in the different stages of the protocol. 

Interaction in our model is driven by the network communication topology
that accounts for radial power supply configurations (i.e, 
tree-like structures where root nodes provide power to the respective subtrees),
and for a notion of proximity (so as to capture, e.g., backup links and address
network reconfigurations). 
We therefore consider that nodes can interact if they are in a provide/receive
power relation or in a neighbouring relation. To this end, synchronisation specifications 
include a direction to determine the target of the synchronisation.
For example, by $\gfilter{\mathsf{Locate}}{\toall}{\_}{\_}$ 
we represent that $\mathsf{Locate}$ targets all (child) nodes ($\toall$) that receive
power from the enabling node, used in the scenario 
by the root(s) to search for the
fault in the power supply (sub)tree. Also, by
$\gfilter{\mathsf{Recover}}{\toup}{\_}{\_}$ we
specify that $\mathsf{Recover}$ targets the power provider ($\toup$)
of the enabling node (i.e., the parent node), used in the scenario to signal when the fault
has been found. Finally, by $\grfilter{\mathsf{Power}}{\_}{\_}$ we represent
that $\mathsf{Power}$ targets a neighbouring node ($\toright$), used in the scenario
for capturing power restoration and associated network reconfiguration. 


Combining two of the above example synchronisations by means of 
recursion and summation ($\plusop$), we may then write a simplified fault 
management protocol 
\begin{equation}\label{ex:basicerp}
\rulename{Simple}\triangleq\rec{(\gbfilter{\mathsf{Locate}}{c_1}{c_1 \vee c_2}\varX \plusop \gufilter{\mathsf{Recover}}{c_2}{\true}\zero)}
\end{equation}
which specifies an alternative between $\mathsf{Locate}$ and 
$\mathsf{Recover}$ synchronisations, where in the former case the protocol 
starts over and in the latter case terminates. Also, to support a more fine grained 
description of protocols, the 
synchronisation actions specify conditions (omitted previously with $\_$) that must hold 
for both enabling and reacting nodes in order for a synchronisation to take place. For 
example, only nodes that are at the root of a (sub)tree that has a fault may enable
a synchronisation on $\mathsf{Locate}$, which is captured in condition $c_1$ (left unspecified here). Also, only nodes that satisfy such condition ($c_1$) or that are without power supply
(which is captured in condition $c_2$) can react to $\mathsf{Locate}$. For 
$\mathsf{Recover}$, the specification says that only nodes that are without power 
supply ($c_2$) can enable the synchronisation, while any node can react to it 
(captured by condition true $\true$). The general idea is that 
the $\mathsf{Locate}$ synchronisations propagate throughout the
power supply tree, one level per synchronisation, up to the point the 
fault is located and synchronisation $\mathsf{Recover}$ leads to termination.

The combination of the summation and of synchronisation 
conditions thus allows for a fine-grained specification of the operation control protocol.
Furthermore, the conditions specified for synchronisations are checked against the state of nodes,
so state information is accounted for in our model. To represent the 
dynamics of systems, such state may evolve throughout the stages of the operation 
protocol. We therefore consider that synchronisation actions may have side-effects
on the state of the nodes that synchronise. This allows to avoid introducing 
specialised primitives and simplifies reasoning on protocols, given the atomicity of the 
synchronisation and side-effect in one step. 




The general principles described above, identified in the context of a fault management 
scenario, guided the design of the language presented in the next section. 
Although this is not the case for a single fault, we remark
that our language addresses scenarios where several nodes may be active
simultaneously in possibly different stages of the protocol. 
A detailed account of the fault management scenario is presented in 
Sect.~\ref{sec:scenario}, including a
 protocol designed to account for configurations with several faults.

\section{A Model for Operation Control Protocols}\label{sec:definition} 

The syntax of the language is given in Table~\ref{tab:syntax}, where
we assume a set of node identifiers ($\id, \ldots$),
synchronisation action labels ($\fun, \ldots$), and logical conditions ($c,i,o,\ldots$).
Protocols ($P,Q, \ldots$) combine static specifications
and the \emph{active} node construct. 
The latter is denoted by $\idop{\id}P$,
representing that node $\id$ is active to carry out the protocol $P$.
Static specifications
include termination $\zero$, fork $P \parop Q$ which says
that both $P$ and $Q$ are to be carried out, and infinite behaviour 
defined in terms of the recursion $\rec{P}$ and recursion variable
$\varX$, with the usual meaning. Finally, static specifications include 
synchronisation summations ($S, \ldots$), where $S_1 \plusop S_2$ says that
either $S_1$ is to be carried out or $S_2$ (exclusively), and where
$\filter{o}{i}P$ represents a synchronisation action: a node 
active on $\filter{o}{i}P$ that satisfies condition $o$ may synchronise
on $\fun$ with the node(s) identified by the direction $d$ for which
condition $i$ holds, leading to the activation of the latter node(s)
on protocol $P$. Intuitively, the node active on $\filter{o}{i}P$
enables the synchronisation, which results in the reaction of the 
targeted nodes that are activated to carry out the continuation protocol $P$.

\begin{table}[tbp]
 \caption{Global Language Syntax}
\centering 
\begin{tabular}{ll}
\medskip
\textrm{(Protocol)}  & 
$P ::=  \zero \quad |\quad  P\parop P \quad |\quad \rec{P} \quad |\quad \varX \quad |\quad S\quad |\quad\idop{\id}P$ \\

\medskip

\textrm{(Summation)}   & $S ::=  \filter{o}{i}P\quad |\quad S\plusop S$ \\

\medskip

\textrm{(Direction)}   & 
$d ::=   \toall\quad |\quad \toup \quad |\quad\toright\quad |\quad\toself $ 
\end{tabular}
\label{tab:syntax}
\end{table}

A direction $d$ specifies the target(s) of a synchronisation action, that 
may be of four kinds: $\toall$ targets all (children) nodes to which the enabling node
provides power to; $\toup$ targets
the power provider (i.e., the parent) of the enabler; $\toright$ targets a neighbour of
the enabler; and $\toself$ targets the enabler itself, used
to capture local computation steps.
We remark on the $\toall$ direction given its particular nature: 
since one node can supply power to several others, synchronisations
with $\toall$ direction may actually involve several reacting nodes, up 
to the respective condition. We interpret $\toall$ synchronisations as
broadcasts, in the sense that we take $\toall$ to target all (direct) child 
nodes that satisfy the reacting condition, which hence comprises the 
empty set (in case the node has no children or none of them satisfy the condition).
The interpretation of binary interaction differs, as synchronisation is only possible
if the identified target node satisfies the condition. 
\begin{example}
\label{ex:constructs}
Consider the following protocol, assuming definitions for $\rulename{Recovery}$, 
$\rulename{Isolation}$ and $\rulename{Restoration}$

\medskip{\centerline{$
\idop{\id}(\gbfilter{\mathsf{Locate}}{o_1}{i_1}\rulename{Recovery} \plusop \gsfilter{\mathsf{End}}{o_2}{i_2}
(\rulename{Isolation}\parop \rulename{Restoration} ))
$}}\medskip

\noindent
which specifies node $\id$ is active to synchronise on $\mathsf{Locate}$ or $\mathsf{End}$,
exclusively. 
\end{example}

Example~\ref{ex:constructs} already hints on the two ways of introducing
concurrency in our model. On the one hand, broadcast can lead to the 
activation of several nodes: in the example, each one of the nodes reacting
to $\mathsf{Locate}$ will carry out 
the $\rulename{Recovery}$ protocol. On the other hand, the fork construct 
allows for a single node to carry out two subprotocols, possibly activating
different nodes in the continuation: in the example, 
a node active to carry out $(\rulename{Isolation}\parop \rulename{Restoration} )$ 
may synchronise with different nodes in each one of the branches.

%
%
%

In order to define the semantics of the language, we introduce structural congruence
that, in particular, captures the relation between the active node construct $\idop{\id}P$ 
and protocol specifications (including the active node construct itself and the fork construct).
Structural congruence is the least 
congruence relation on protocols that satisfies the rules given in Table~\ref{tab:struct}. 
The first set of rules captures expected principles, namely that fork
and summation are associative and commutative, and that fork has identity element $\zero$
(we remark that the syntax excludes $\zero$ as a branch in summations).
Rule  $\idop{\id}(P \parop Q)\equiv  \idop{\id}P \parop \idop{\id}Q$ captures the 
interpretation of the fork construct: it is equivalent to specify that a node $\id$ is active 
on a fork, and to specify that a fork has $\id$ active on both branches. Rule 
$\idop{\id_1}\idop{\id_2}P \equiv  \idop{\id_2}\idop{\id_1}P $ says that active
nodes can be permuted and rule $\idop{\id}\zero\equiv  \zero$ says that a node
active to carry out termination is equivalent to the termination itself. Structural congruence 
together with reduction define the operational semantics of the model. Intuitively, 
structural congruence rewriting allows active nodes to ``float'' in the term 
towards the synchronisation actions. 

\begin{example}
Considering the active node distribution in a fork,
we have that

\medskip
$
\gbfilter{\mathsf{Locate}}{o_1}{i_1}\rulename{Recovery} \plusop \gsfilter{\mathsf{End}}{o_2}{i_2}
\idop{\id}(\rulename{Isolation}\parop \rulename{Restoration} )
$

\medskip\noindent
is structural congruent to

\medskip
$
\gbfilter{\mathsf{Locate}}{o_1}{i_1}\rulename{Recovery} \plusop \gsfilter{\mathsf{End}}{o_2}{i_2}
(\idop{\id}\rulename{Isolation}\parop \idop{\id}\rulename{Restoration} )
$
\end{example}

%
\begin{table}[t]
 \caption{Structural Congruence}
\centering 
\begin{tabular}{ccccc}
$P \parop \zero \equiv P$
&\hspace{0.2cm} &
$P_1 \parop (P_2 \parop P_3) \equiv (P_1 \parop P_2) \parop P_3$
& \hspace{0.3cm} &
$P_1 \parop P_2 \equiv P_2 \parop P_1$ 
\\[0.2cm]
 $\rec{P}\equiv  P[\rec{P}/\varX]$
 &  & 
 $S_1 \plusop (S_2 \plusop S_3) \equiv (S_1 \plusop S_2) \plusop S_3$
 &  & 
$S_1 \plusop S_2 \equiv S_2 \plusop S_1$
\\[0.2cm]
 $\idop{\id}(P \parop Q)\equiv  \idop{\id}P \parop \idop{\id}Q$ 
&&
$\idop{\id_1}\idop{\id_2}P \equiv  \idop{\id_2}\idop{\id_1}P $ 
&&
 $\idop{\id}\zero\equiv  \zero$ 
\end{tabular}
\label{tab:struct}
\end{table}
The definition of reduction depends on the network topology and on the fact that
nodes satisfy certain logical conditions. We consider state information for each 
node so as to capture both ``local'' information about the topology (such as the 
identities of the power provider and of the set of neighbours) and other information 
relevant for condition assessment (such as the status of the power supply). 
The network state, denoted by $\Delta$, is a mapping from node identifiers to states, 
where a state, denoted by $\state$, is a register
$\station{\id}{\source{\id'}{\sourcestate}}{\nei}{\capa}{\act}{\err}$
containing the following information:
$\id$ is the node identifier; $\id'$ identifies the power provider; $\sourcestate$
captures the status of the input power connection; $\nei$ is the set of identifiers of neighbouring 
nodes; $\capa$ is the power supply capacity of the node; $\act$ is the number of 
active power supply links (i.e., the number of nodes that receive power from this one); 
and $\err$ is the number of power supply links that are in a faulty state.

As mentioned in Section~\ref{sec:preview}, we check conditions against states for the
purpose of allowing synchronisations.
Given a state $\state$ we denote by $\state \models c$ that state $s$ satisfies condition $c$,
where we leave the underlying logic unspecified.
For example, we may say that $\state \models (\capa > 0)$ to check
that $\state$ has capacity greater than $0$.

Also mentioned in Section~\ref{sec:preview} is the notion of side-effects, in the sense 
that synchronisation actions may result in state changes so as to model system evolution.
By $\upd(\id,\id',\func{d}{},\Delta)$ we denote the operation that yields the network
state obtained by updating $\Delta$ considering node $\id$ synchronises on
$\fun$ with $\id'$, hence the update regards the side-effects of $\fun$
in the involved nodes. Namely, given $\Delta =( \Delta', \id \mapsto \state, \id' \mapsto \state')$ 
we have that $\upd(\id,\id',\func{d}{},\Delta)$ is defined as
$(\Delta', \id \mapsto \func{d}{}!(\state, \id'), \id' \mapsto \func{d}{}?(\state',\id))$,
where $ \func{d}{}!(\state, \id')$ modifies state $\state$ according to the side-effects of 
enabling $\func{d}{}$ and considering $\id'$ is the reactive node (likewise for the reacting 
update, distinguished by $?$). We consider side-effects only for binary
synchronisations ($\toright$ and $\toup$ directions), 
but state changes could also be considered for other directions in similar lines.

The definition of reduction relies on an auxiliary operation, denoted $d(\Delta,\id)$, 
that yields the recipient(s) of 
a synchronisation action, given the direction $d$, the network state $\Delta$, and the 
enabler of the action $\id$. The operation, defined as follows, thus yields 
the power provider of the node in case the direction is $\toup$, (any) one of the neighbours 
in case the direction is $\toright$, all the nodes that have as parent the enabler in case the 
direction is $\toall$, and is undefined for direction $\toself$.

\smallskip{\centerline{$
\begin{array}{llll}
\toup(\Delta, \id) & \deff & \id' \quad (\mbox{if}\; \Delta(\id) = 
\station{\id}{\source{\id'}{\sourcestate}}{\nei}{\capa}{\act}{\err})
\\
\toright(\Delta, \id) & \deff & \id' \quad (\mbox{if}\; \Delta(\id) = 
\station{\id}{\source{\id''}{\sourcestate}}{\nei}{\capa}{\act}{\err} \;\mbox{and}\; 
\id' \in \nei)
\\
\toall(\Delta,\id) & \deff & \{ \id' \ | \ \toup(\Delta,\id') = \id \}
\\
\toself(\Delta,\id) &\deff& \mbox{\it undefined}
\end{array}
$}}\smallskip

The reduction relation is given in terms of configurations 
consisting
of a protocol $P$ and a mapping $\Delta$, for which we use $\sys{\Delta}{P}$ to denote the combination. 
By $\sys{\Delta}{P} \red \sys{\Delta'}{P'}$ we represent that configuration $\sys{\Delta}{P}$ evolves
in one step 
to configuration $\sys{\Delta'}{P'}$, potentially involving state changes ($\Delta$ and $\Delta'$ may differ) and
(necessarily) involving a step in the protocol from $P$ to $P'$. 
\\ \indent
\begin{table}[t]
\caption{Reduction Rules}
\centering 
\begin{tabular}{c}
$
\infer[(\rulename{Bin})]{
\sys{\Delta}{\idop{\id}(\filter{o}{i}P\plusop S)}\red
\sys{\Delta'}{\filter{o}{i}(\idop{{\id'}}P)\plusop S}
}{
d\in\{\toup,\toright\}\quad  \Delta(\id)\models o\quad 
d(\Delta,\id)=\id' \quad \Delta(\id')\models i \qquad\Delta'=\upd(\id,\id',\func{d}{},\Delta)
}$
\\[0.3cm]
$
\infer[(\rulename{Brd})]{
\sys{\Delta}{\idop{\id}(\bfilter{o}{i}P \plusop S)}\red
\sys{\Delta}{\bfilter{o}{i}(\idop{{\tilde{\ids}}}P) \plusop S}
}{
\Delta(\id)\models o
\quad \toall(\Delta,\id)= \ids' \quad 
\ids = \{\id' \ | \ \id' \in \ids' \wedge \Delta(\id') \models i \}
}$
\\[0.3cm]
$
\infer[(\rulename{Loc})]{
\sys{\Delta}{\idop{\id}(\sfilter{o}{i}P\plusop S)}\red
\sys{\Delta}{\sfilter{o}{i}(\idop{\id}P)\plusop S}
}{
\Delta(\id)\models o
\quad \Delta(\id)\models i 
}$
\\[0.3cm]
$
\infer[(\rulename{Synch})]{
\sys{\Delta}{\filter{o}{i}P}\red
\sys{\Delta'}{\filter{o}{i}P'}
}{
\sys{\Delta}{P}\red
\sys{\Delta'}{P'}
}$\qquad

$
\infer[(\rulename{Id})]{
\sys{\Delta}{\idop{\id}P}\red
\sys{\Delta'}{\idop{\id}P'}
}{
\sys{\Delta}{P}\red
\sys{\Delta'}{P'}
}$\qquad

\\[0.3cm]

$
\infer[(\rulename{Sum})]{
\sys{\Delta}{P_1\plusop P_2}\red
\sys{\Delta'}{P'_1\plusop P_2}
}{
\sys{\Delta}{P_1}\red
\sys{\Delta'}{P'_1}
}$\qquad

$
\infer[(\rulename{Par})]{
\sys{\Delta}{P_1\parop P_2}\red
\sys{\Delta'}{P'_1\parop P_2}
}{
\sys{\Delta}{P_1}\red
\sys{\Delta'}{P'_1}
}$

\\[0.3cm]
%

$
\infer[(\rulename{Struct})]{
\sys{\Delta}{P}\red
\sys{\Delta'}{Q}
}{
P\equiv P'\qquad \sys{\Delta}{P'}\red
\sys{\Delta'}{Q'}\qquad Q'\equiv Q
}$

\end{tabular}
\label{tab:red}
\end{table}
Reduction
is defined as the least relation that satisfies the rules shown in Table~\ref{tab:red},
briefly described next.
Rule $\rulename{Bin}$ captures the case of binary interaction, hence when 
the direction ($d$) of the synchronisation action targets either the parent ($\toup$) or  
a neighbour ($\toright$). Protocol $\idop{\id}(\filter{o}{i}P\plusop S)$ says that node
$\id$ is active on a synchronisation summation that includes $\filter{o}{i}P$, so 
$\id$ can enable a synchronisation on $\fun$ provided that the state of $\id$ satisfies 
condition $o$, as specified in premise $\Delta(\id) \models o$. Furthermore,
the reacting node $\id'$, specified in the permise $d(\Delta,\id)=\id'$,
is required to satisfy condition $i$. In such circumstances, the configuration
can evolve and the resulting state is obtained by updating the states of
the involved nodes considering the side-effects of the synchronisation,
and where the resulting protocol $\filter{o}{i}(\idop{{\id'}}P)\plusop S$ specifies that  
 $\id'$ is active on the continuation protocol $P$. Notice that the synchronisation 
 action itself is preserved (we return to this point when addressing the language
 closure rules).
\\ \indent
Rule $\rulename{Brd}$ captures the case of broadcast interaction ($\toall$), following lines
similar to $\rulename{Bin}$. Apart from the absence of state update, the main difference 
is that now a set of potential reacting nodes
is identified ($\ids'$ denotes a set of node identifiers), 
out of which all those satisfying condition $i$ are singled out ($\ids$).
The latter are activated in the continuation protocol, to represent which we use 
$\idop{\tilde{\ids}}$ to
abbreviate $\idop{\id_1}\ldots\idop{\id_m}$ considering $\ids = \id_1,\ldots, \id_m$. We remark
that the set of actual reacting nodes may be empty (e.g., if none of the potential ones
satisfies condition $i$), in which case $(\tilde{\emptyset})P$ is defined as $P$. The 
fact that the reduction step nevertheless takes place, even without actual reacting nodes, 
motivates the choice of the broadcast terminology,
and differs from the binary interaction which is blocked while the targeted 
reacting node does not satisfy the condition.
\\ \indent
Rule $\rulename{Loc}$ captures the case of local computation steps ($\toself$). For the sake
of uniformity we keep (both) output and input conditions that must be met by the
state of the active node. Notice that the node that carries out the $\fun$ step retains control, 
i.e., the same $\id$ is active before and after the synchronisation on $\fun$. Like for broadcast,
we consider local steps do not involve any state update.

Rules for protocol language closure yield the interpretation that nodes can be active
at any stage of the protocol, hence reduction may take place at any level, including
after a synchronisation action (rule $\rulename{Synch}$) and 
within a summation (rule $\rulename{Sum}$). By preserving 
the structure of the protocol, including synchronisation actions that have been carried 
out, we allow for 
participants to be active on (exactly) the same stage of the protocol simultaneously and
at different moments in time.
Rules $\rulename{Id}$ and $\rulename{Par}$ follow the same principle that reduction 
can take place at any point in the protocol term, while accounting for state changes.
Rule $\rulename{Struct}$ closes reduction under structural congruence, so as to
allow the reduction rules to focus on the case of an active node 
that immediately scopes over a synchronisation action summation.

Since we are interested in developing protocols that may be used in different
networks, we consider such development is to be carried out using what we 
call static protocols, i.e., protocols that do
not include the active node construct. Then, to represent a concrete operating 
system, active nodes may be added at ``top-level'' to the static specification 
(e.g., $\idop{\id}\rulename{Recovery}$ where $\rulename{Recovery}$ does 
not include any active nodes), together with the network state.
Also, for the purpose of simplifying protocol design, we consider 
that action labels are unique (up to recursion unfolding). 
As usual, we exclude protocols where recursion appears unguarded by at least
one synchronisation action (e.g., $\rec{\varX}$).

Such notions allow us to streamline the distributed implementation presented 
next. From now on, we only consider well-formed protocols that follow the above
guidelines, i.e., originate from specifications where the active node construct only appears top-level,
all action labels are distinct, and recursion is guarded.

\section{Distributed Implementation}\label{sec:corress}
In this section, we present the automatic translation of global language 
specifications considering a target model where the global 
protocol descriptions are carried out in a distributed way. Namely, we
synthesise the controllers that operate locally in each node from
the global specification, and ensure the correctness of the translation
by means of an operational correspondence result. The development
presented here can therefore be seen as a proof of concept that global
descriptions may be automatically compiled to provably correct 
distributed implementations.

We start by presenting the target model, which consists of a network
of nodes where each node has a state $\state$ (previously introduced) 
and a behaviour given in terms of \emph{definitions}, \emph{reactions}
and \emph{choices}. Intuitively, definitions allow nodes to synchronise
on actions, after which proceeding as specified in reactions. The latter
are defined as alternative behaviours specified in choices.
The syntax of behaviours and networks
is shown in Table~\ref{tab:dsyntax}, reusing syntactic elements
given in the global language (namely 
action labels $\fun$, directions $d$, and conditions $c$).

\begin{table}[t]
 \caption{The Syntax of the Distributed Language}
\centering 
\begin{tabular}{llll}

\smallskip

\textrm{(Definition)} & $D$ & $::=$  &  $\mconfig{c}\fun^{d}?.{R} \quad |\quad R \quad |\quad D\parop D $ \\

\smallskip

\textrm{(Reaction)} & $R$ & $::=$ & $ C \quad |\quad \zero\quad|\quad R \parop R $ \\

\smallskip

\textrm{(Choice)} & $C$ & $::=$ & $\mconfig{c}\fun^{d}! \quad |\quad C \plusop C $ \\

\smallskip

\textrm{(Network)}  & $\net$ & $::=$ &$ \node{\state}{D}\quad |\quad \net \nparop \net$ \\
\end{tabular}
\label{tab:dsyntax}
\end{table}

A definition $D$ may either be a pair of simultaneously active (sub)definitions
 $D_1\parop D_2$, a reaction $R$, or the (persistent) input 
$ \mconfig{c}\fun^{d}?.{R}$. The latter allows a node to react to a synchronisation
on $\fun$ (according to direction $d$), provided that the node satisfies condition $c$,
leading to the activation of reaction $R$. A reaction $R$ can either be a pair of 
simultaneously active (sub)reactions $R_1\parop R_2$ (that can be specified as 
continuation of an input), the inaction $\zero$ or a choice $C$. The latter 
is either a pair of (sub)choices $C_1 \plusop C_2$ or the output 
$\mconfig{c}\fun^{d}!$ that
allows a node to enable a synchronisation on $\fun$ (targeting the nodes
specified by direction $d$), provided that the node satisfies condition $c$.

A network $\net$ is either a pair of (sub)networks 
$\net_1 \nparop \net_2$ or a node $\node{\state}{D}$ which comprises
a behaviour given by a definition ($D$) and a state ($\state$). We recall that
a state is a register of the form $\station{\id}{\source{\id'}{\sourcestate}}{\nei}{\capa}{\act}{\err}$.

The operational semantics of networks is defined by a labelled transition system (LTS), 
which relies on the operational semantics of {definitions}, also defined by an LTS.
We denote by $D_1 \rTo{\alpha} D_2$ that definition $D_1$ exhibits action $\alpha$
and evolves to $D_2$. The actions ranged over by $\alpha$ are 
$\mconfig{c}\fun^{d}?$, $\mconfig{c}\fun^{d}!$, and $\mconfig{c}\fun$.
Action $\mconfig{c}\fun^{d}?$ represents the ability to react to a synchronisation on $\fun$, 
provided that the node satisfies condition $c$, and according to direction $d$. 
Similarly, $\mconfig{c}\fun^{d}!$ represents the ability to enable a synchronisation 
on $\fun$, also considering the condition $c$ and the targeting direction $d$.
A local computation step is captured by $\mconfig{c}\fun$, which also specifies
the label of the action $\fun$ and a condition $c$. The rules that define the LTS 
of {definitions}, briefly explained next, are shown in Table~\ref{tab:behsem}.



\begin{table}[tbp]
\caption{Semantics of Definitions}
\centering 
\begin{tabular}{c}

$
 \mconfig{c}\fun^{d}?.R \rTo{\mconfig{c}\fun^{d}?}  R \parop  \mconfig{c}\fun^{d}?.R \quad\rulename{Inp}
$\qquad\qquad

$
\mconfig{c}\fun^{d}! \rTo{\mconfig{c}\fun^{d}!}  \zero \quad \rulename{Out}
$

\\[0.3cm]

$
\infer[\rulename{Sum}]{
C_1\plusop C_2 \rTo{\alpha} C_1' 
}{C_1\rTo{\alpha} C_1'}$\quad

$
\infer[\rulename{Int}]{
D_1 \parop D_2 \rTo{\alpha} D'_1 \parop D_2
}{D_1\rTo{\alpha} D_1'}$

$\
\infer[\rulename{Self}]{
D_1 \parop D_2 \rTo{\mconfig{c_1\land c_2}\fun} D_1' \parop D_2'
}{D_1\rTo{\mconfig{c_1}\fun^{\toself}!} D_1'
\quad D_2\rTo{\mconfig{c_2}\fun^{\toself}?} D_2'}$
\\[0.3cm]


\end{tabular}
\label{tab:behsem}
\end{table}

Rule $\rulename{Inp}$ states that an input can  
exhibit the corresponding reactive transition, comprising synchronisation
action label $\fun$, condition $c$, and direction $d$. The input results in
the activation of the respective reaction $R$, while the input itself is preserved,
which means that all inputs are persistently available. Likewise, rule 
$\rulename{Out}$ states that an output can exhibit the corresponding
synchronisation enabling transition, after which it terminates.
Rules $\rulename{Sum}$ and $\rulename{Int}$ are standard, 
specifying alternative and interleaving behaviour, respectively.
Rule $\rulename{Self}$ captures local computation steps, 
which are actually the result of a synchronisation between an
output (reaction) and an input (definition) which specify 
direction $\toself$. We remark that both conditions are registered
in the action label of the conclusion (by means of a conjunction),
so the node must satisfy them in order for the computation 
step to be carried out. 
For the sake of brevity, we omit the symmetrical rules of \rulename{Sum}, \rulename{Int}, and \rulename{Self}.
%
%
%
%
%
%
%
%

The operational semantics that defines the LTS of networks is shown in Table~\ref{tab:netsem},
where we denote by $\net_1 \rTo{\lambda} \net_2$ that network $\net_1$ exhibits 
label $\lambda$ and evolves to network $\net_2$. The transition labels ranged over by $\lambda$ 
are $\linktr{\link}{\fun}$ and $\tau$, capturing network interactions and local computation steps. 
The label $ \linktr{\link}{\fun}$ comprises the action label $\fun$ and the communication link $\link$
which represents either binary 
($\id\rightarrow \id$ and $\id\leftarrow\id$) or broadcast ($\id!\toall$ and $\id?\toall$) interaction.
\\
\indent
The binary link $\id_1\rightarrow \id_2$ specifies that node $\id_1$ is willing to enable a synchronisation with node $\id_2$ while $\id_1\leftarrow \id_2$ specifies that node $\id_1$ is willing to react to a synchronisation from node $\id_2$. Furthermore, the broadcast link $\id!\toall$ specifies that node $\id$ is willing to enable a synchronisation with all of its direct children, while $\id?\toall$ specifies that a node is willing to react to 
a broadcast from node $\id$. We use $\id(\state)$ and $i(\state)$ to denote the identities of the node itself and 
of the parent (i.e., if $\state = \station{\id_1}{\source{\id_2}{\sourcestate}}{\nei}{\capa}{\act}{\err}$
then $\id(\state) = \id_1$ and $i(\state) = \id_2$). 
\\
\indent
We consider short-range broadcast in the sense that only direct children of the node $\id$ can be the 
target of the synchronisation on $\linktr{\id!\toall}{\fun}$ (i.e., only nodes that satisfy condition $i(s)=\id$). Moreover, a node can accept a synchronisation from its parent only: (1) if its local definitions are able to react 
to the synchronisation action (i.e., $D\rTo{\mconfig{c}\fun^{\toall}?} D'$ is defined); also, (2) if the current state 
of the node satisfies the condition of the defined reaction (i.e., $D\rTo{\mconfig{c}\fun^{\toall}?} D'$ 
and $\state\models c$). We define a predicate
$\mathit{discard}(\node{\state}{D},\linktr{\id?\toall}{\fun})$ that takes as parameters a node 
$\node{\state}{D}$ and a network label $\linktr{\id?\toall}{\fun}$. This predicate is used to identify 
the case when nodes may discard broadcasts, i.e., when any of the above conditions 
is not satisfied.  

%


\begin{table}[tbp]
	\caption{Semantics of Networks}
\centering 
\begin{tabular}{c}
$
\infer[\rulename{Loc}]{
\node{\state}{D}\rTo{\tau} \node{\state}{D'}
}{D\rTo{\mconfig{c}\fun} D'\quad
\state\models c
}$\qquad
$
\infer[\rulename{oBinU}]{
\node{\state}{D}\rTo{\linktr{\id(\state)\rightarrow i(\state)}{\fun}} \node{\state'}{D'}
}{D\rTo{\mconfig{c}\fun^{\toup}!} D'\quad
\state\models c\quad \state'=\fun^\toup!(\state,i(\state))
}$
\\[0.3cm]
$
\infer[\rulename{oBinR}]{
\node{\state}{D}\rTo{\linktr{\id(\state)\rightarrow \id}{\fun}} \node{\state'}{D'}
}{D\rTo{\mconfig{c}\fun^{\toright}!} D'\quad
\state\models c\quad \state'=\fun^\toright!(\state,\id)\quad \id\in \nei(\state)
}$
\\[0.3cm]
$
\infer[\rulename{iBin}]{
\node{\state}{D}\rTo{\linktr{\id(\state)\leftarrow \id}{\fun}} \node{\state'}{D'}
}{d\in\mset{\toup,\toright}\quad D\rTo{\mconfig{c}\fun^{d}?} D'\quad \state\models c\quad
\state'=\fun^{d}?(\state,\id)
}$
\\[0.3cm]

$
\infer[\rulename{oBrd}]{
\node{\state}{D}\rTo{\linktr{\id(\state)!\toall}{\fun}} \node{\state}{D'}
}{
D\rTo{\mconfig{c}\fun^{\toall}!} D'\; \state\models c
}$
\;
$
\infer[\rulename{iBrd}]{
\node{\state}{ D} \rTo{\linktr{i(\state)?\toall}{\fun}} \node{\state}{D'}
}{D\rTo{\mconfig{c}\fun^{\toall}?} D'\; \state\models c\;
}$
\;

$
\infer[\rulename{dBrd}]{
\node{\state} { D}\rTo{\linktr{\id?\toall}{\fun}} \node{\state}{D}
}
{ \mathit{discard}(\node{\state}{D},\linktr{\id?\toall}{\fun})}$
\\[0.3cm]

$
\infer[\rulename{Com}]{
\net_1 \nparop \net_2\rTo{\gamma(\lambda_1,\lambda_2)}\net'_1 \nparop \net'_2
}{
\net_1\rTo{\lambda_1}\net'_1 \quad \net_2\rTo{\lambda_2}\net'_2
}$
\quad
$
\infer[\rulename{Par}]{
\net_1 \nparop \net_2\rTo{\lambda}\net'_1 \nparop \net_2
}{
\net_1\rTo{\lambda}\net'_1 \quad \lambda\not\in\{\linktr{\id!\toall}{\fun},\ \linktr{\id?\toall}{\fun}\}
}$
\\[0.3cm]
\end{tabular}
\label{tab:netsem}
\end{table}

We briefly describe the rules given in Table~\ref{tab:netsem}, which address
individual nodes and networks. The former, roughly, lift the LTS of 
definitions to the level of nodes taking into account conditions and side-effects
of synchronisation actions. 

Rule \rulename{Loc} states that a node can evolve silently with a $\tau$ transition when its definitions exhibit a local computation step $\mconfig{c}\fun$, provided that the state of the node satisfies the condition of the computation step $\state\models c$. 
Rule \rulename{oBinU} is used to synchronise with the parent node, rule \rulename{oBinR} is used to 
synchronise with a neighbour and rule \rulename{iBin} is used to react to a synchronisation from either a 
child or a neighbour node. More precisely, rules $\rulename{oBinU}$/$\rulename{oBinR}$ and
$\rulename{iBinU}$ express that nodes can respectively exhibit enabling/reactive
transitions provided that local definitions in $D$ exhibit the corresponding 
transitions, with synchronisation action label $\fun$, direction $\toup$ or $\toright$ and condition $c$. 
Also, the condition $c$ is checked against the state $\state$, and the latter
is updated according to the side-effects of the synchronisation (for both
enabling and reacting nodes, as described previously for the global language semantics).
We remark that the rules for binary interaction register in the communication
link the identities of the interacting parties. 

Rule \rulename{oBrd} states that a node can enable a broadcast synchronisation on action $\fun$ if local definitions in $D$ can exhibit the corresponding enabling transition $D\rTo{\mconfig{c}\fun^{\toall}!} D'$, and condition $c$ is satisfied by the local state $\state\models c$. Similarly, rule \rulename{iBrd} states that a node can react to a broadcast synchronisation from its parent only when local definitions in $D$
can exhibit the corresponding transition
$D\rTo{\mconfig{c}\fun^{\toall}?} D'$ and the condition is satisfied by the local state $\state\models c$. 
Otherwise, rule \rulename{dBrd} may be applied, capturing the case when the node discards the 
synchronisation and remains unaltered. 

Rule \rulename{Par} (and its omitted symmetric) governs the interleaving of networks when binary and local actions are observed, i.e, $\lambda\not\in\{\linktr{\id!\toall}{\fun},\ \linktr{\id?\toall}{\fun}\}$. Rule \rulename{Com} governs the synchronisation of networks when either binary or broadcast actions are observed. The function $\gamma(\lambda_1,\lambda_2)$ identifies the resulting label, and is defined as follows. 

\[
\gamma(\lambda_1,\lambda_2)=\left\{
\begin{array}{ll}
\tau & \qquad \text{if}\; \lambda_1=\linktr{\id_1\rightarrow\id_2}{\fun}\ \&\ \lambda_2=\linktr{\id_2\leftarrow\id_1}{\fun}\\
\tau & \qquad \text{if}\; \lambda_1=\linktr{\id_1\leftarrow\id_2}{\fun}\ \&\ \lambda_2=\linktr{\id_2\rightarrow\id_1}{\fun}\\
\lambda_1 & \qquad \text{if}\; \lambda_1=\linktr{\id!\toall}{\fun}\ \&\ \lambda_2=\linktr{\id?\toall}{\fun}\\
\lambda_2 & \qquad \text{if}\; \lambda_1=\linktr{\id?\toall}{\fun}\ \&\ \lambda_2=\linktr{\id!\toall}{\fun}\\
\lambda & \qquad \text{if}\; \lambda_1=\lambda_2=\lambda=\linktr{\id?\toall}{\fun}\\
\bot & \qquad \mbox{otherwise}
\end{array}
\right.
\]

The function returns $\tau$ when synchronisation on a binary action is possible and the enabling label when a synchronisation on a broadcast action is possible, as reported in the first four cases respectively. The fifth case states that both networks can react simultaneously to the same broadcast, otherwise the function is undefined as reported in the last case. We remark that the semantics of broadcasts is presented in a standard way
(cf.~\cite{cbs,forte}).

\subsection{Local Controller Synthesis}
\begin{table}[t]
\caption{Projection}\label{tab:proj}
\[
\begin{array}{lcllll}
\proj{\filter{o}{i}P}{?}{\sigma} & \deff &  \mconfig{i}\fun^{d}?.{\proj{P}{!}{\sigma}} \parop \proj{P}{?}{\sigma} 
& &&\rulename{iSynch}\\

\proj{\filter{o}{i}P}{!}{\sigma} & \deff & \mconfig{o}\fun^{d}! & &&\rulename{oSynch}\\ 

\proj{\filter{o}{i}P}{\id}{\sigma} & \deff & \proj{P}{\id}{\sigma} & &&\rulename{idSynch}
\\
\\
\proj{\idop{\id}P}{\id}{\sigma} & \deff & \proj{P}{\id}{\sigma} \parop \proj{P}{!}{\sigma} & &&\rulename{idNode}\\


\proj{\idop{\id}P}{\pid}{\sigma} & \deff & \proj{P}{\pid}{\sigma} 
 & (\pid \neq \id)\qquad\qquad & &\rulename{pNode}\\ 

\\
\proj{\zero}{\pid}{\sigma} & \deff & \zero & &&\rulename{pNil}\\
\proj{\varX}{\pid}{\sigma}  & \deff & \zero & (\pid \neq \,!) &&\rulename{pVar}
\\
\proj{\varX}{!}{\sigma}  & \deff & \proj{P}{!}{\sigma} & (\sigma(\varX) = P) &&\rulename{oVar}
\\
\proj{\rec{P}}{\pid}{\sigma} & \deff & \proj{P}{\pid}{\sigma[\varX \mapsto P]} & &&\rulename{pRec}\\


\proj{P \parop Q}{\pid}{\sigma} & \deff & \proj{P}{\pid}{\sigma} \parop \proj{Q}{\pid}{\sigma}& &&\rulename{pPar}\\
 
\proj{S_1\plusop S_2}{!}{\sigma} & \deff & \proj{S_1}{!}{\sigma} \plusop \proj{S_2}{!}{\sigma} 
& &&\rulename{oSum}\\

\proj{S_1\plusop S_2}{\pid}{\sigma} & \deff & \proj{S_1}{\pid}{\sigma} \parop \proj{S_2}{\pid}{\sigma} 
& (\pid \neq \,!)  & &\rulename{pSum}
\\
\\
\tproj{\sys{\Delta}{P}} & \deff & \Pi_{\forall \id \in \mathit{dom}(\Delta)} 
(\node{\Delta(\id)}{
\proj{P}{?}{\emptyset} \parop \proj{P}{\id}{\emptyset}}) & &&\rulename{Proj}

\end{array}
\]
\end{table}

We now turn to the translation of global specifications into the target distributed model.
Intuitively, for each synchronisation action specified in the source protocol, there will be
a corresponding reactive definition in every node. This way, we ensure that all nodes are
equipped to react to any synchronisation. Also, in every active node in the 
global specification there will be a corresponding enabling behaviour. 
Since the global model prescribes the transference of the active role in
synchronisations, the distributed implementation will also specify that reactions of
definitions lead to the activation of enabling behaviour. 
\\
\indent
The projection function, defined by the cases reported in Table~\ref{tab:proj}, realises the 
above principles, in particular in the three types of projection (ranged over by $\pid$).
Reactive projection, denoted by $?$, is used to generate the persistent inputs that capture reactive behaviour;
active projection, denoted by $\id$ (for some $\id$), is used to identify and generate the enabling behaviours 
of node $\id$;
and enabling projection, denoted by $!$, is used to generate the outputs that capture enabling behaviour.
We then denote by $\proj{P}{\pid}{\sigma}$ the projection of the protocol $P$ according to 
parameters $\pid$ and $\sigma$. The latter is a mapping from recursion variables to protocols, 
used to memorise recursive protocols and abstract from their unfolding. 
\\
\indent
We briefly explain the cases in Table~\ref{tab:proj}. 
Case $\rulename{iSynch}$ shows the reactive projection of the synchronisation action,
yielding a persistent input with the respective condition $i$, label $\fun$, and direction $d$.
The continuation of the input is obtained by the enabling projection ($!$) of the continuation
of the synchronisation action ($P$). Hence, the reaction leads to the enabling of the continuation
which captures the transference of the active role in synchronisations. The result of
the projection also specifies the (simultaneously active $\parop$) reactive projection of the 
continuation protocol, so as to generate the corresponding persistent inputs.
Case $\rulename{oSynch}$ shows the enabling projection of the synchronisation action,
yielding the output considering the respective condition $o$, label $\fun$, and direction $d$.
Notice that the fact that outputs do not specify continuations is aligned with the idea that 
a node yields the active role after enabling a synchronisation. 

Case $\rulename{idSynch}$ shows the active projection of the synchronisation action,
yielding the active projection of the continuation. The idea is that active projection 
inspects the structure of the protocol, and introduces enabling behaviour whenever
an active node construct specifying the respective $\id$ is found (i.e., $\idop{\id}P$). 
This is made precise 
in case $\rulename{idNode}$, where the projection yields both the enabling projection 
of the protocol $P$ together with its the active projection, so as to address configurations
in which the same node is active in different stages of the protocol. Case $\rulename{pNode}$
instead shows that the other types of projection of the active node construct result in the 
respective projection of the continuation.

Case $\rulename{pNil}$ says that the 
terminated protocol is projected (in all types of projection) to inaction ($\zero$). 
Case $\rulename{pVar}$ says that the active/reactive projections of the recursion variable 
also yield inaction, hence do not require reasoning on the unfolding. 
Instead, case $\rulename{oVar}$ shows the enabling projection of the recursion variable,
yielding the enabling projection of the variable mapping, which allows to account for the
unfolding. This is made precise in case $\rulename{pRec}$ where the mapping is updated
with the association of the variable maps to the recursion body.

We remark that well-formed protocols do not specify active node constructs in the body
of a recursion, since they originate from protocols where all active node constructs are top-level, 
hence the active projection of any recursive protocol necessarily 
yields inaction (nevertheless captured in the general case). 

Case $\rulename{pPar}$ says that (all types of) projection of the fork protocol
yields the (simultaneously active) respective projections of the branches of the fork.
Cases $\rulename{oSum}$ and $\rulename{pSum}$ address the summation protocol:
on the one hand, the enabling projection yields the choice between the projections
of the branches; on the other hand, reactive/active projection yield the
simultaneously active projections of the branches. Notice that choices may only
specify outputs, hence only enabling projection may yield alternative (output) behaviour.
Reactive projection yields a collection of persistent inputs (one 
per synchronisation action) which are simultaneously active.
Active projection generates the enabling behaviour of active nodes, so if such 
active nodes are found in (continuations of) the branches of the summation, 
their enabling behaviour is taken as simultaneously active.

The projection of a configuration, denoted $\tproj{\sys{\Delta}{P}}$ and 
defined in the $\rulename{Proj}$ case, specifies a parallel composition of all
nodes of the network (i.e., all those comprised in the network state). Each 
node is obtained by considering the state yielded by the respective network 
state mapping, and considering the behaviour is yielded by a combination of 
the reactive and the active projections of the protocol $P$. Notice that the
active projection is carried out considering the node identifier, hence the 
result potentially differs between distinct nodes, while the reactive projection
is exactly the same for all nodes. Intuitively, consider the reactive projection
as the static collection of reactive definitions, 
and the active projection as the runtime (immediately 
available) enabling behaviour.
\begin{example} The reactive projection of the \rulename{Simple} protocol in Sect.~\ref{sec:preview}(\ref{ex:basicerp}) is 
\[
\proj{\rulename{Simple}}{?}{\emptyset} = \mconfig{c_1\lor c_2}\mathsf{Locate}^{\star}?.\left(\mconfig{c_1}\mathsf{Locate}^{\star}!\plusop \mconfig{c_2}\mathsf{Recover}^{\toup}!\right) \parop \mconfig{\true}\mathsf{Recover}^{\toup}?.0
\]
\end{example}

\begin{table}[t]
 \caption{Structural Congruence - Definitions and Networks}
\centering 
\begin{tabular}{ccccc}
\\
$D \parop \zero \equiv D$
&\hspace{0.2cm} &
$D_1 \parop (D_2 \parop D_3) \equiv (D_1 \parop D_2) \parop D_3$
& \hspace{0.3cm} &
$D_1 \parop D_2 \equiv D_2 \parop D_1$ 
\\[0.2cm]
$\mconfig{c}\fun^{d}?.{R} \parop \mconfig{c}\fun^{d}?.{R} \equiv \mconfig{c}\fun^{d}?.{R}$
& &
 $D_1 \plusop (D_2 \plusop D_3) \equiv (D_1 \plusop D_2) \plusop D_3$
 &  & 
$D_1 \plusop D_2 \equiv D_2 \plusop D_1$
\\\\
$D_1 \equiv D_2 \Rightarrow  \node{\state}{D_1} \equiv  \node{\state}{D_2}$
& &
$\net_1 \nparop (\net_2 \nparop \net_3) \equiv (\net_1 \nparop \net_2) \nparop \net_3$
& &
$\net_1 \nparop \net_2 \equiv \net_2 \nparop \net_1$ 
\end{tabular}
\label{tab:structnets}
\end{table}

For the purpose of our operational correspondence result, we consider
structural congruence of networks and of behaviours defined by the rules shown in Table~\ref{tab:structnets}. The rules capture expected
principles (namely, that operators $\nparop$, $\parop$, and $\plusop$ are 
associative and commutative, and that $\parop$ has identity element $\zero$)
and an absorbing principle for persistent inputs (i.e., $ \mconfig{i}\fun^{d}?.R \parop 
 \mconfig{i}\fun^{d}?.R \equiv  \mconfig{i}\fun^{d}?.R$), which allows to reason
 about persistent inputs as if they are unique, necessary when considering the
 reacting projection of a recursive protocol and its unfolding. We may now 
 state our operational correspondence result, where we denote by
 $\net_1 \rTo{\lambda}{\!\!_\equiv\;} \net_2$ that there exists $\net'$ such that
 $\net_1 \rTo{\lambda}{} \net'$ and $\net' \equiv \net_2$, and where we use
$\hat{\lambda}$ to range over $\linktr{\id!\toall}{\fun}$ and  
$\tau$ labels.

\begin{theorem}[Operational Correspondence]
\label{the:opcorr}
\begin{enumerate}
\item[]
\item If $\sys{\Delta}{P} \rTo{}{} \sys{\Delta'}{Q}$ then 
$\tproj{\sys{\Delta}{P}} \rTo{\hat{\lambda}}{\!\!_\equiv\;} \tproj{\sys{\Delta'}{Q}}$.
\item If $\tproj{\sys{\Delta}{P}} \rTo{\hat{\lambda}}{} \net$ then
$\sys{\Delta}{P} \rTo{}{} \sys{\Delta'}{Q}$ and $\net \equiv \tproj{\sys{\Delta'}{Q}}$.
\end{enumerate}
\end{theorem}

Theorem~\ref{the:opcorr} states a one-to-one correspondence in the behaviours
of global specifications and their translations. Notice that the possible transition
labels ($\hat{\lambda}$) exclusively refer to synchronisations, either broadcast 
or binary, and local actions. 

The operational correspondence result is key to allow for the development of protocols
to be carried out in the global language, as any reasoning for the global model 
can be conveyed to the (automatically generated)  implementation.

\section{Fault Management in Power Distribution Grids}\label{sec:scenario}
\begin{figure}[t]
\caption{Power Distribution Grid}
\begin{center}
\includegraphics[scale=.22]{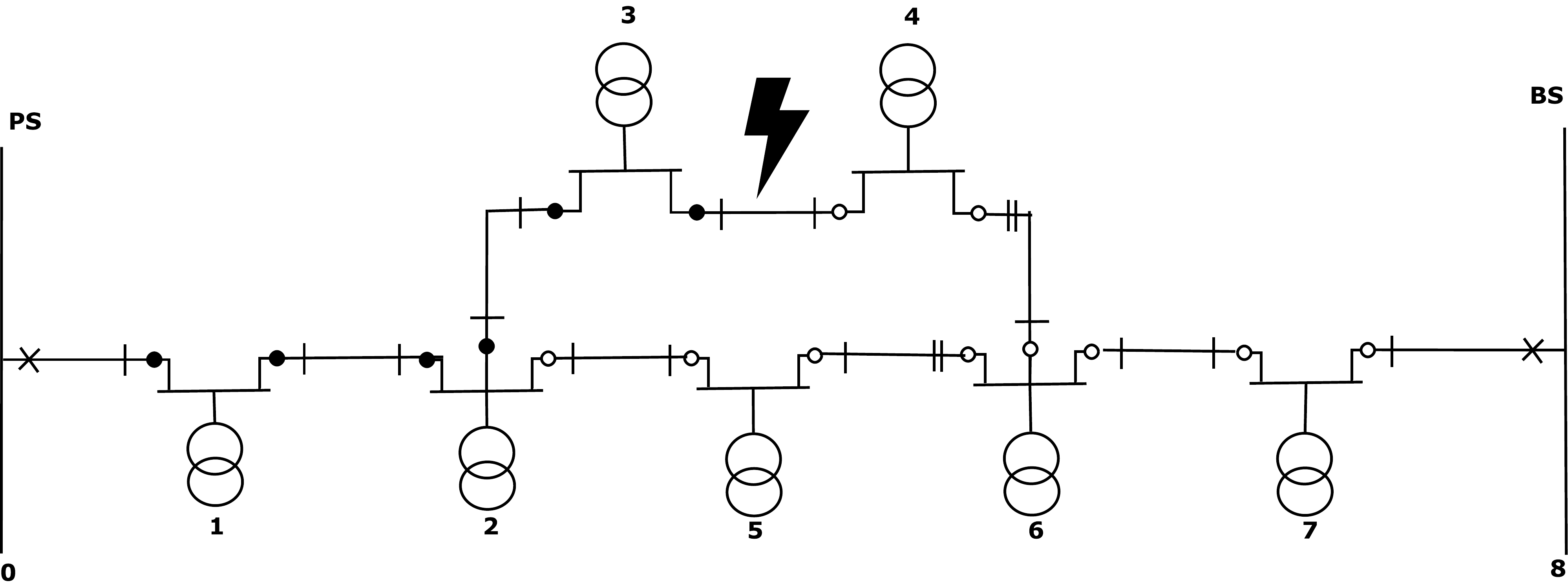} 
\end{center}
\label{fig:pdn}
\end{figure}
In this section, we model a non trivial fault management protocol that handles error detection, localisation, and isolation in power grids. Furthermore, we model autonomous network reconfiguration and power restoration. 
 We specify the behaviour of our protocol from a global point of view (i.e., using the global language) which is more natural and less error-prone
with respect to considering the local viewpoints.  

We consider a cross section of a power grid's network as reported in Fig.~\ref{fig:pdn}. The cross section shows a radial power grid's network with a primary power substation \textbf{PS},  a backup power substation \textbf{BS}, and seven secondary power substations, numbered from $\bf1$ to $\bf 7$. The type of this network is called radial because every substation has only one incoming power input and possibly multiple power outputs. In some sense, the  substation can be viewed as a power router.

Primary and backup substations have circuit breakers $\times$ that open up when a fault occurs in their domain, e.g., a transmission power line breaks. Each secondary substation has fault indicators (fault $\bullet$ and no fault $\circ$), line switches (closed $|$ and open $\|$), and an embedded controller that implements the substation's behaviour and manages interactions with others. Fig.~\ref{fig:pdn} illustrates a configuration where the secondary substations $\bf 1$-$\bf 5$ are energised by the primary substation $\bf PS$, while secondary substations $\bf 6$ and $\bf 7$ are energised by the backup substation $\bf BS$. Secondary substations cannot operate the switches or exchange information without authorisation from the primary substation which supplies the power. 
\\
\indent
Let us consider that a fault occurs in the domain of the primary substation $\bf PS$, e.g., the transmission power line between substations $\bf 3$ and $\bf 4$ breaks. The primary substation can sense the existence of a fault because its circuit breaker $\times$ opens up, but it cannot determine the location of the fault. Hence, the substation $\bf PS$ initiates the fault recovery protocol by synchronising with its directly connected secondary substations and delegates them to activate the local fault recovery protocol. The delegation between secondary substations propagates in the direction given by their fault indicators. 
\\
\indent
A secondary substation first validates the error signal by measuring its voltage level. If the voltage is zero, the secondary substation activates the fault recovery protocol, otherwise the signal is discarded. 
Once a secondary substation activates the fault recovery protocol, it inspects its own fault indicators and determines whether the fault is on its input or output power lines. If the fault is on its output, it delegates the substations connected to its output power lines to collaborate to locate the fault. If the fault is on its input power line, the substation (in our case,  substation $\bf 4$) takes control and initiates both isolation and power restoration. The former consists of isolating the faulty line and restoring the power to the network's segment located before the fault, while the latter consists of restoring the power to the network's segment located after the fault.
\\
\indent
For isolation, substation $\bf 4$ opens its switches to the faulty line segment and collaborates with $\bf 3$ to open its affected line switches as well. Notice that now the domain of the primary substation $\bf PS$ is segmented into two islands: $\{\bf PS,1,2,3,5\}$ and $\{\bf 4\}$. The control is transferred back to the primary substation in a step by step fashion and the power is restored to the first island. For power restoration, substation $\bf 4$ asks for power supply from one of its neighbours that is capable of supplying an additional substation (in our case, substation $\bf 6$). Once substation $\bf 6$ supplies power to $\bf 4$, the latter changes its power source to $\bf 6$ and now substation $\bf 4$ belongs to the domain of the backup substation $\bf BS$.
\\
\indent
We now show how to provide a simple and intuitive global specification for our scenario using the linguistic primitives of the global language. In what follows, we use the following terminology: the state of a source link $\sourcestate$ can be $0$ (to indicate a faulty link) or $1$ otherwise. We will use $z$ in place of the source $\id$ when a substation is not connected to a power supply. The initial state of each substation follows from Fig.~\ref{fig:pdn}. For instance, substations $\bf 3$, $\bf 4$, and $\bf 6$  have the following initial states $\station{\bf 3}{\source{2}{1}}{\{2,4\}}{1}{1}{1}$, $\station{\bf 4}{\source{3}{0}}{\{3,6\}}{1}{0}{0}$, and $\station{\bf 6}{\source{7}{1}}{\{4,5,7\}}{2}{0}{0}$, respectively.
The recovery protocol is reported below. 

\medskip{\centerline{$
\begin{array}{rcl}
\rulename{Recovery} & \triangleq & \rec{(\gbfilter{\mathsf{Locate}}{o_1}{i_1}\varX \plusop \gsfilter{\mathsf{End}}{o_2}{i_2}\rulename{Islanding}) }\\[1ex]
\rulename{Islanding} & \triangleq & \rulename{IsolationStart}\parop \rulename{Restoration} 
\end{array}
$}}\medskip

The protocol states that either $\mathsf{Locate}$ is broadcasted to the children of the enabling substation, after which the protocol starts over, or $\mathsf{End}$ is carried out and the node retains control. In case $\mathsf{End}$ is carried out, the enabling substation proceeds by activating the \rulename{Islanding} protocol. The latter specifies a fork between \rulename{IsolationStart} and the \rulename{Restoration} protocols. 

A substation enabled on $\rulename{Recovery}$ can broadcast $\mathsf{Locate}$ only when it has at least one faulty output, i.e.,  $o_1=(\err>0)$. Furthermore, receiving substations can synchronise on $\mathsf{Locate}$ only if they have fault on their output or their input, i.e., $i_1=(\err>0)\lor (t=0)$. On the other hand, $\mathsf{End}$ can be carried out when the enabling substation has a fault on its input, i.e., $o_2=(t=0)$ and {$i_2=\true$}. Both actions have no side-effects on states.

The \rulename{Isolation} and the \rulename{Restoration} protocols are reported below:

\medskip{\centerline{$
\begin{array}{rcl}
\rulename{IsolationStart} & \triangleq & \gufilter{\mathsf{Recover}}{o_3}{i_3}\zero 
\plusop 
\gufilter{\mathsf{RecoverDone}}{o_3}{i_4} \rulename{Isolation}
\\[1ex]
\rulename{Isolation} & \triangleq &
\rec{(\gufilter{\mathsf{Isolate}}{o_5}{i_3}\zero \plusop \gufilter{\mathsf{IsolateDone}}{o_5}{i_4}\varX \plusop \gsfilter{\mathsf{Stop}}{o_6}{i_2}\zero)})
\\[1ex]
\rulename{Restoration} & \triangleq & \grfilter{\mathsf{Power}}{o_7}{i_7}\zero \\
\end{array}
$}}\medskip

The \rulename{IsolationStart} protocol states that a synchronisation between the enabling substation and its parent on the $\mathsf{Recover}$ or $\mathsf{RecoverDone}$ actions may happen only if the enabling substation has a fault on its input, i.e., $o_3=(t=0)$. The parent reacts on $\mathsf{Recover}$ if the number of faults on its output is greater than one, i.e., $i_3=(\err>1)$. In this case the parent does not proceed, since there are still more faulty links to be handled. Instead, the $\mathsf{RecoverDone}$ synchronisation captures the case for the handling of the last faulty link, i.e., $i_4=(\err=1)$, in which case the parent takes control and proceeds to $\rulename{Isolation}$. In both cases, the enabling station disconnects itself from the faulty line (setting parent to $z$) and the parent decrements its output faults ($\err$), its active outputs ($\act$), and also its capacity ($\capa$). Furthermore, both enabling and parent node remove each other from the list of neighbours, thus isolating the faulty link.
\\ \indent
When the parent is enabled on $\rulename{Isolation}$ then three synchronisation branches on the $\mathsf{Isolate}$, the $\mathsf{IsolateDone}$ and the $\mathsf{Stop}$ actions can be carried out. An enabling substation can synchronise with its parent on an $\mathsf{Isolate}$ or $\mathsf{IsolateDone}$ actions if it is not the primary station, i.e., $o_5=(i\neq\infty)$. As a side-effect on the state of the parent, the number of faults is decremented ($\err$). Notice that the interpretation of the two branches on $\mathsf{Isolate}$ and $\mathsf{IsolateDone}$ is similar to the one on $\mathsf{Recover}$ and $\mathsf{RecoverDone}$. Only the primary station ($o_6=(i=\infty)$) can execute the $\mathsf{Stop}$ action, ending the \rulename{Isolation} protocol. As mentioned before, local actions have no side-effects. 
\\ \indent
Finally, the \rulename{Restoration} protocol states that a disconnected substation, i.e., $o_7=(i=z)$, can synchronise with one of its neighbours ($\toright$) on the $\mathsf{Power}$ action as a request for a power supply. Only a neighbour that has enough power and is fully functional ($i_7=(\capa>\act\land \err=0)$) can engage in the synchronisation. By doing so, the neighbour increments its active outputs ($\act$) and the enabling substation marks its neighbour as its power source. 
The protocol then terminates after the network reconfiguration. Note that although the \rulename{IsolationStart} and the \rulename{Restoration} protocols are specified in parallel, a specific order is actually induced by the synchronisation conditions. Thanks to the output condition $o_7=(i=z)$ of the $\mathsf{Power}$ action, we are sure that the power supply can be reestablished only after the faulty link is disconnected, which is a side-effect of either $\mathsf{Recover}$ or $\mathsf{RecoverDone}$ actions. Thus, the synchronisation on $\mathsf{Power}$ can only happen after the synchronisation on either 
$\mathsf{Recover}$ or $\mathsf{RecoverDone}$.
\\ \indent
The static protocol \rulename{Recovery} abstracts from the concrete network configuration. To represent a concrete grid network, active substations must be added at \qt{top-level} to the \rulename{Recovery} protocol, together with the network state, i.e., $\sys{\Delta}{\idop{\bf PS}\rulename{Recovery}}$ where $\Delta$ is a mapping from a substation identifiers to states. Notice that the primary station, $\bf PS$, is initially active because, according to our scenario, it is the only station that has rights to initiate protocols. 
\\ \indent
We remark that the protocol is designed so as to handle configurations with multiple faults, where several nodes may be active simultaneously on, e.g., 
$\mathsf{Locate}$ albeit belonging to different parts (subtrees) of the network.
\section{Concluding Remarks}
\label{sec:related}
We propose a model of interaction for power distribution grids. More precisely, we show how to specify operation control protocols governing the behaviour of a power gird as a whole from a global perspective. We formalise how global model specifications can be used to automatically synthesise individual controllers of the grid substations, yielding an operationally correct distributed implementation. We also show how to use our global language to model a non-trivial fault management scenario from the realm of power girds. 
Notice that the design principles of our model target power distribution grids specifically,
namely considering that interaction is confined to closely follow the network
topology. However, the principles showcased by our development can be
used when considering other topologies. Conceivably, our principles can also
be conveyed to other settings where operation protocols involve
yielding control as a consequence of synchronisations.
We conclude this paper by relating to existing literature, focusing on formal models in particular, and by mentioning some directions for future work.

Global protocol specifications can be found in the session type literature, spawning 
from the work of Honda et al.~\cite{HYC08}. Such specifications are typically
used to verify programs or guide their development. More recently, we find
proposals that enrich the global protocol specifications so as to allow programming
to be carried out directly in the global language (e.g., choreographies~\cite{CM13}).
We insert our development in this context, since we also provide a global specification
language where operation control protocols can be programmed. The distinguishing
features of our approach naturally stem from the targeted setting. In particular, we consider 
that protocol specifications are role-agnostic and that the interacting parties are established 
operationally in the following sense: emitting/enabling parties are determined from active 
nodes and receiving/reacting parties are determined by the network topology, and
the active role is transferred in synchronisations. 
Previous session-type based approaches typically consider protocols include role 
annotations, which are impersonated by parties at runtime when agreeing to collaborate 
on a protocol. Some works also consider role-agnostic protocol specifications~\cite{CV10},
that roles have a more flexible interpretation, for instance that the association between party
and role is dynamic~\cite{BCVV12} and that several parties may simultaneously impersonate
a single role~\cite{DY11}.
None of such works, as far as we know, considers that
the receiver is determined by the combination of who is the sender and the network 
topology, and that control is yielded in communication.

For future work, it is definitely of interest to develop verification techniques for system level safety properties,
and it seems natural to specify and check such properties considering the global model. 
To reason on protocol correctness and certify operation on all possible configurations, such verification techniques 
should focus on the static specification and abstract from the concrete network configuration, relying in some form of abstract interpretation~\cite{cousot77,cousot79}. 
Thanks to the operational correspondence result that relates the global and distributed models, we can ensure that any 
property that holds for the global specification also holds for the distributed one. The reactive style of the distributed
model suggests an implementation based on the actor model~\cite{actors} (in, e.g., Erlang~\cite{erlang}). Conceivably,
such reactive descriptions can be useful in other realms, in particular when
addressing interaction models based on immediate reactions to external stimulus.
\appendix
\appendix 
\renewcommand{\thesection}{A}

\section{Operational Correspondence}
In this section, we sketch the proof of operational correspondence, separating the proof of soundness 
(Lemma~\ref{lem:sound}) and completeness (Lemma~\ref{lem:completeness}) in two subsections 
(Section~\ref{sec:soundness} and Section~\ref{sec:completeness}, respectively).

\subsection{Operational Correspondence - Soundness}
\label{sec:soundness}

We introduce auxiliary lemmas to structure the proof. 
The following lemma ensures that structurally equivalent networks have equivalent behaviours.
\begin{lemma}[Definition and Network LTS Closure Under Structural Congruence] \label{lem:congstruct}
\begin{enumerate}
\item[]{}
\item
If $D_1 \rTo{{\alpha}}{} D_1'$ and $D_1 \equiv D_2$ then $D_2 \rTo{{\alpha}}{\!\!_\equiv\;} D_1'$.
\item
If $\net_1 \rTo{{\lambda}}{} \net_1'$ and $\net_1 \equiv \net_2$ then $\net_2 \rTo{{\lambda}}{\!\!_\equiv\;} \net_1'$.
\end{enumerate}
\end{lemma}
\begin{proof}
The proof follows by induction on the length of the derivation $D_1 \equiv D_2$ and $\net_1 \equiv \net_2$ in expected lines,
where $\mathit{2}$ relies on $\mathit{1}$.
\end{proof}


The following results are used in the proof that structural congruence is preserved by projection, in particular regarding recursion unfolding.
We start by addressing properties of the $!$-projection,
namely, Lemma~\ref{lem:AC} shows that it is preserved under any environment and process substitution 
since only the initial actions are relevant.

\begin{lemma}[Preservation of $!$-Projection]
\label{lem:AC}
Let $P$ be a protocol where recursion is guarded. We have that 
\begin{enumerate}
\item $\proj{P}{!}{\sigma} \equiv \proj{P}{!}{\sigma'}$ for any $\sigma,\sigma'$. 
\item $\proj{P[Q/\varX]}{!}{\sigma} \equiv \proj{P}{!}{\sigma}$.
\end{enumerate}
\end{lemma}
\begin{proof}
By induction on the structure of $P$ following expected lines. Notice that $!$-projection addresses only immediate synchronisation
actions, hence neither the mapping $\sigma$ nor the substitution affect the projection given that recursion is guarded in $P$.
\end{proof}

Lemma~\ref{lem:AC} is used directly in the proof of recursion unfolding for the $!$-projection and also to prove properties 
of the $?$-projection.
Lemma~\ref{lem:D} is also auxiliary to the case of $?$-projection, showing the correspondence between environment and 
process substitution used in $!$-projection.

\begin{lemma}[Soundness of Mapping for $!$-Projection]
\label{lem:D}
Let $Q$ be a protocol where recursion in guarded. We have that $\proj{P}{!}{\sigma[\varX \mapsto Q]} \equiv \proj{P[Q/\varX]}{!}{\sigma}$.
\end{lemma}
\begin{proof}
By induction on the structure of $P$ following expected lines. Notice that when $P$ is $\varX$ then we obtain
$\proj{Q}{!}{\sigma[\varX \mapsto Q]}$ and $\proj{Q}{!}{\sigma}$ which may be equated considering Lemma~\ref{lem:AC}(\emph{1}).
\end{proof}

We now address properties that directly regard the $?$-projection, starting by Lemma~\ref{lem:B} that shows that processes
in the environment are interchangeable as long as their $!$-projection is equivalent.

\begin{lemma}[Preservation of $?$-Projection] 
\label{lem:B}
Let $Q_1$ and $Q_2$ be protocols where recursion is guarded such that $\proj{Q_1}{!}{\sigma} \equiv \proj{Q_2}{!}{\sigma}$.
We have that $\proj{P}{?}{\sigma[\varX \mapsto Q_1]} \equiv \proj{P}{?}{\sigma[\varX \mapsto Q_2]}$.
\end{lemma}
\begin{proof}
By induction on the structure of $P$ following expected lines. Notice that when $P$ is $\varX$ then we obtain
$\proj{Q_1}{!}{\sigma[\varX \mapsto Q_1]}$ and $\proj{Q_2}{!}{\sigma[\varX \mapsto Q_2]}$
which may be equated considering Lemma~\ref{lem:AC}(\emph{1}) and the hypothesis.
\end{proof}

The following result (Lemma~\ref{lem:E}) is key to the proof of the unfolding, as it shows that two copies of a $?$-projection 
are equivalente to one.

\begin{lemma}[Replicability of $?$-Projection] 
\label{lem:E}
We have that $\proj{P}{?}{\sigma} \parop \proj{P}{?}{\sigma} \equiv \proj{P}{?}{\sigma}$.
\end{lemma}
\begin{proof}
By induction on the structure of $P$. The proof follows by induction hypothesis in expected lines for all cases except for
the synchronisation action which also relies on axiom $\mconfig{c}\fun^{d}?.{R} \parop \mconfig{c}\fun^{d}?.{R} \equiv \mconfig{c}\fun^{d}?.{R}$.
\end{proof}

The main auxiliary properties for the case of recursion unfolding regarding $?$-projection are given by the next result.
Lemma~\ref{lem:FH}(\emph{1}) equates the substitution with the environment, considering the involved process is in context.
Lemma~\ref{lem:FH}(\emph{2}) equates the substitution with the environment, considering a context with some (input) residua.

\begin{lemma}[Soundness of Mapping for $?$-Projection]
\label{lem:FH}
Let $Q$ be a protocol where recursion is guarded.
\begin{enumerate}
\item 
We have that $\proj{P}{?}{\sigma[\varX \mapsto Q]} \parop \proj{Q}{?}{\sigma} \equiv  \proj{P[Q/\varX]}{?}{\sigma} \parop \proj{Q}{?}{\sigma} $.
\item
We have that $ \proj{P[Q/\varX]}{?}{\sigma} \equiv  \proj{P}{?}{\sigma[\varX \mapsto Q]} \parop \Pi_{i \in I} \mconfig{c_i}\fun_{i}^{d_i}?.{R_i}$.
\end{enumerate}
\end{lemma}
\begin{proof}
By induction on the structure of $P$. We sketch the proof of \emph{1.}, the proof of \emph{2.} follows similar lines.
\begin{description}
\item[Case $P$ is $\zero$] We have that 
$\zero \parop \proj{Q}{?}{\sigma} \equiv  \zero \parop \proj{Q}{?}{\sigma} $.

\item[Case $P$ is $\varX$] Since $\proj{\varX}{?}{\sigma[\varX \mapsto Q]}$ by definition is $\zero$, we have that 
$\zero \parop \proj{Q}{?}{\sigma} \equiv  \proj{Q}{?}{\sigma} \parop \proj{Q}{?}{\sigma} $ and the proof follows
by considering Lemma~\ref{lem:E}.

\item[Case $P$ is $\idop{\id}P'$ or $P_1 \parop P_2$ or $S_1 \plusop S_2$] The proof follows from the induction hypothesis in expected lines.

\item[Case $P$ is $\filter{o}{i}P'$] Since $\proj{\filter{o}{i}P'}{?}{\sigma[\varX \mapsto Q]}$ by definition is 
$\mconfig{i}\fun^{d}?.{\proj{P'}{!}{\sigma[\varX \mapsto Q]}} \parop \proj{P'}{?}{\sigma[\varX \mapsto Q]}$
we have that (\emph{i}) $\proj{\filter{o}{i}P'}{?}{\sigma[\varX \mapsto Q]} \parop \proj{Q}{?}{\sigma} \equiv
\mconfig{i}\fun^{d}?.{\proj{P'}{!}{\sigma[\varX \mapsto Q]}} \parop \proj{P'}{?}{\sigma[\varX \mapsto Q]} \parop \proj{Q}{?}{\sigma}$.
Considering Lemma~\ref{lem:D} we have that (\emph{ii}) $\mconfig{i}\fun^{d}?.{\proj{P'}{!}{\sigma[\varX \mapsto Q]}}
\equiv \mconfig{i}\fun^{d}?.{\proj{P'[Q/\varX]}{!}{\sigma}}$.
By induction hypothesis we conclude that (\emph{iii}) $ \proj{P'}{?}{\sigma[\varX \mapsto Q]} \parop  \proj{Q}{?}{\sigma} \equiv 
\proj{P'[Q/\varX]}{?}{\sigma} \parop \proj{Q}{?}{\sigma} $. Considering both (\emph{ii}) and (\emph{iii}) we conclude
that (\emph{iv})
$$\mconfig{i}\fun^{d}?.{\proj{P'}{!}{\sigma[\varX \mapsto Q]}} \parop \proj{P'}{?}{\sigma[\varX \mapsto Q]} \parop \proj{Q}{?}{\sigma}
\equiv
\mconfig{i}\fun^{d}?.{\proj{P'[Q/\varX]}{!}{\sigma}} \parop \proj{P'[Q/\varX]}{?}{\sigma} \parop \proj{Q}{?}{\sigma}$$
By definition we have that
$\proj{\filter{o}{i}P'[Q/\varX]}{?}{\sigma} \parop \proj{Q}{?}{\sigma}
\equiv \mconfig{i}\fun^{d}?.{\proj{P'[Q/\varX]}{!}{\sigma}} \parop \proj{P'[Q/\varX]}{?}{\sigma} \parop \proj{Q}{?}{\sigma}$
which considering (\emph{i}) and (\emph{iv}) completes the proof.

\item[Case $P$ is $\rec{P'}$] By definition we have that 
$\proj{\rec{P'}}{?}{\sigma[\varX' \mapsto Q]} \parop \proj{Q}{?}{\sigma} \equiv
\proj{P'}{?}{\sigma[\varX' \mapsto Q][\varX \mapsto P']} \parop \proj{Q}{?}{\sigma}$. By induction hypothesis
we have that $\proj{P'}{?}{\sigma[\varX' \mapsto Q][\varX \mapsto P']} \parop \proj{Q}{?}{\sigma}
\equiv \proj{P'[Q/\varX']}{?}{\sigma[\varX \mapsto P']} \parop \proj{Q}{?}{\sigma}$. 
Considering Lemma~\ref{lem:AC}(\emph{2}) since recursion is guarded in $P'$ we have that 
$\proj{P'[Q/X']}{!}{\sigma} \equiv \proj{P'}{!}{\sigma}$, from which, considering Lemma~\ref{lem:B}, we conclude
$\proj{P'[Q/\varX']}{?}{\sigma[\varX \mapsto P']} \parop \proj{Q}{?}{\sigma}\equiv
\proj{P'[Q/\varX']}{?}{\sigma[\varX \mapsto (P'[Q / \varX'])]} \parop \proj{Q}{?}{\sigma}$.
The latter concludes the proof since by definition 
$\proj{\rec{P'}[Q/\varX]}{?}{\sigma} \parop \proj{Q}{?}{\sigma}
\equiv \proj{P'[Q/\varX']}{?}{\sigma[\varX \mapsto (P'[Q / \varX'])]} \parop \proj{Q}{?}{\sigma}$.
\end{description}
\end{proof}


%

The following result is used in the proof of the $\id$-projection.

\begin{lemma}[Static Protocol $\id$-Projection]
\label{lem:idabsent}
Let $P$ be an $\idop{id}$-absent protocol. We have that $\proj{P}{\id}{\sigma} \equiv \zero$.
\end{lemma}
\begin{proof}
By induction on the structure of $P$ following expected lines.
\end{proof}

We may now prove that projection is preserved under recursion unfolding. This lemma will be used in the proof of Lemma \ref{lem:structP}, where we want to show that the projection function is invariant to structural congruent protocols. 

\begin{lemma}[Preservation of Projection under Unfolding]\label{lem:recdel}
Let $P$ be an $\idop{id}$-absent protocol where recursion is guarded. We have that
$\proj{\rec{P}}{r}{\emptyset} \equiv \proj{P[\rec{P}/\varX]}{r}{\emptyset}$ for any $r$.
\end{lemma}
\begin{proof} We prove the three cases separately.

\begin{description}
\item[Case $r = \;!$] 
By definition we have $\proj{\rec{P}}{!}{\emptyset} \deff \proj{P}{!}{[\varX \mapsto P]}$.
Since recursion is guarded in $P$, from Lemma~\ref{lem:AC}(\emph{2}) we conclude $\proj{P}{!}{[\varX \mapsto P]} \equiv 
\proj{P[\rec{P}/\varX}{!}{[\varX \mapsto P]}$, and from Lemma~\ref{lem:AC}(\emph{1}) we have 
$\proj{P[\rec{P}/\varX}{!}{[\varX \mapsto P]} \equiv \proj{P[\rec{P}/\varX}{!}{\emptyset}$.

\item[Case $r = \id$] 
Since $P$ is an $\idop{id}$-absent protocol we have that $P[\rec{P}/\varX]$ is an $\idop{id}$-absent protocol, 
hence the result follows immediately from Lemma~\ref{lem:idabsent} from which we conclude
$\proj{\rec{P}}{\id}{\emptyset} \equiv \zero$
and $\proj{P[\rec{P}/\varX]}{\id}{\emptyset} \equiv \zero$.

\item[Case $r = \;?$]
From Lemma~\ref{lem:FH}(\emph{1}) we have that 
$$\proj{P}{?}{[\varX \mapsto \rec{P}]} \parop \proj{\rec{P}}{?}{\emptyset} \equiv
\proj{P[\rec{P}/\varX]}{?}{\emptyset} \parop \proj{\rec{P}}{?}{\emptyset}$$
By definition we have that $\proj{\rec{P}}{!}{\emptyset} \equiv \proj{P}{!}{[\varX \mapsto P]}$ 
and since recursion is guarded in $P$ from Lemma~\ref{lem:AC}(\emph{1}) we 
have that $\proj{P}{!}{[\varX \mapsto P]} \equiv \proj{P}{!}{\emptyset}$ hence
$\proj{\rec{P}}{!}{\emptyset} \equiv \proj{P}{!}{\emptyset}$.
From this fact, considering Lemma~\ref{lem:B} we conclude
$$\proj{P}{?}{[\varX \mapsto P]} \parop \proj{\rec{P}}{?}{\emptyset} \equiv
\proj{P[\rec{P}/\varX]}{?}{\emptyset} \parop \proj{\rec{P}}{?}{\emptyset}$$
By definition we have that $\proj{\rec{P}}{?}{\emptyset} \equiv \proj{P}{?}{[\varX \mapsto P]}$, hence
$$\proj{\rec{P}}{?}{\emptyset} \parop \proj{\rec{P}}{?}{\emptyset} \equiv
\proj{P[\rec{P}/\varX]}{?}{\emptyset} \parop \proj{\rec{P}}{?}{\emptyset}$$
From Lemma~\ref{lem:E} we have that $\proj{\rec{P}}{?}{\emptyset} \parop \proj{\rec{P}}{?}{\emptyset}\equiv \proj{\rec{P}}{?}{\emptyset}$, hence 
(\emph{i})
$$ \proj{\rec{P}}{?}{\emptyset} \equiv
\proj{P[\rec{P}/\varX]}{?}{\emptyset} \parop \proj{\rec{P}}{?}{\emptyset}$$ 

From Lemma~\ref{lem:FH}(\emph{2}) we have that
$$ \proj{P[\rec{P}/\varX]}{?}{\sigma} \equiv  \proj{P}{?}{\sigma[\varX \mapsto \rec{P}]} \parop \Pi_{l \in L} \mconfig{c_l}\fun_{l}^{d_l}?.{R_l}$$
As before, from $\proj{\rec{P}}{!}{\emptyset} \equiv \proj{P}{!}{\emptyset}$ and considering Lemma~\ref{lem:B} we conclude
$$ \proj{P[\rec{P}/\varX]}{?}{\sigma} \equiv  \proj{P}{?}{\sigma[\varX \mapsto P]} \parop \Pi_{l \in L} \mconfig{c_l}\fun_{l}^{d_l}?.{R_l}$$
and, again as before, since by definition $\proj{\rec{P}}{?}{\emptyset} \equiv \proj{P}{?}{[\varX \mapsto P]}$ we conclude (\emph{ii})
$$ \proj{P[\rec{P}/\varX]}{?}{\sigma} \equiv  \proj{\rec{P}}{?}{\emptyset}\parop \Pi_{l \in L} \mconfig{c_l}\fun_{l}^{d_l}?.{R_l} $$ 
which together with (\emph{i}) 
$$ \proj{\rec{P}}{?}{\emptyset} \equiv
\proj{P[\rec{P}/\varX]}{?}{\emptyset} \parop \proj{\rec{P}}{?}{\emptyset}$$ 
allows us to conclude
$$ \proj{\rec{P}}{?}{\emptyset} \equiv
 \proj{\rec{P}}{?}{\emptyset}\parop \Pi_{l \in L} \mconfig{c_l}\fun_{l}^{d_l}?.{R_l} \parop \proj{\rec{P}}{?}{\emptyset}$$
From this fact and considering Lemma~\ref{lem:E} we conclude
$$ \proj{\rec{P}}{?}{\emptyset} \equiv
 \proj{\rec{P}}{?}{\emptyset}\parop \Pi_{l \in L} \mconfig{c_l}\fun_{l}^{d_l}?.{R_l} $$
which together with (\emph{ii})
 $$ \proj{P[\rec{P}/\varX]}{?}{\sigma} \equiv  \proj{\rec{P}}{?}{\emptyset}\parop \Pi_{l \in L} \mconfig{c_l}\fun_{l}^{d_l}?.{R_l}$$ 
 completes the proof.
\end{description}
\end{proof}

Lemma \ref{lem:structP} states that structurally congruent protocols have structurally congruent projections, namely equivalent reactive projection $\qt{?}$, active projection $\qt{\id}$ and output projection $\qt{!}$. This lemma will be used in the proof of the soundness lemma, Lemma \ref{lem:sound}.
\begin{lemma}[Preservation of Projection Under Structural Congruence]\label{lem:structP}
If $P \equiv Q$ then $\proj{P}{r}{\emptyset} \equiv \proj{Q}{r}{\emptyset}$ for any $r$.
\end{lemma}
\begin{proof}
The proof proceeds by induction on the length of the derivation $P\equiv Q$.
\begin{description}
\item[Case $P$ is $Q\parop \zero$] We need to prove that  $\proj{Q\parop\zero}{r}{\emptyset} \equiv \proj{Q}{r}{\emptyset}$ for any $r$. We have three cases depending on $r$.

\begin{description}
\item[Case $r=?$:] By definition of the projection function, we have that 
$$\proj{Q\parop\zero}{?}{\emptyset}=\proj{Q}{?}{\emptyset}\parop\proj{\zero}{?}{\emptyset}$$

By the definition again, we know that $\proj{\zero}{?}{\emptyset}=\zero$ and now we have that $\proj{Q\parop\zero}{?}{\emptyset} \equiv \proj{Q}{?}{\emptyset}\parop\zero\equiv\proj{Q}{?}{\emptyset}$ as required.
\item[Case $r=!$:] By definition of the projection function, we have that 
$$\proj{Q\parop\zero}{!}{\emptyset}=\proj{Q}{!}{\emptyset}\parop\proj{\zero}{!}{\emptyset}$$

By the definition again, we know that $\proj{\zero}{!}{\emptyset}=\zero$ and now we have that $\proj{Q\parop\zero}{!}{\emptyset} \equiv \proj{Q}{!}{\emptyset}\parop\zero\equiv\proj{Q}{!}{\emptyset}$ as required.
\item[Case $r=\id$:] By definition of the projection function, we have that 
$$\proj{Q\parop\zero}{id}{\emptyset}=\proj{Q}{id}{\emptyset}\parop\proj{\zero}{id}{\emptyset}$$

By the definition again, we know that $\proj{\zero}{id}{\emptyset}=\zero$ and now we have that $\proj{Q\parop\zero}{id}{\emptyset} \equiv \proj{Q}{id}{\emptyset}\parop\zero\equiv\proj{Q}{id}{\emptyset}$ as required.
\end{description}
\item[Case $P$ is $Q_1\parop (Q_2\parop Q_3)$] We need to prove that  $\proj{Q_1\parop (Q_2\parop Q_3)}{r}{\emptyset} \equiv \proj{(Q_1\parop Q_2)\parop Q_3}{r}{\emptyset}$ for any $r$. We have three cases depending on $r$.

\begin{description}
\item[Case $r=?$:] By definition of the projection function, we have that 
$$\proj{Q_1\parop (Q_2\parop Q_3)}{?}{\emptyset}=\proj{Q_1}{?}{\emptyset}\parop\proj{Q_2\parop Q_3}{?}{\emptyset}=\proj{Q_1}{?}{\emptyset}\parop\proj{Q_2}{?}{\emptyset}\parop\proj{Q_3}{?}{\emptyset}$$ and

$$\proj{(Q_1\parop Q_2)\parop Q_3}{?}{\emptyset}=\proj{Q_1\parop Q_2}{?}{\emptyset}\parop\proj{Q_3}{?}{\emptyset}=\proj{Q_1}{?}{\emptyset}\parop\proj{Q_2}{?}{\emptyset}\parop\proj{Q_3}{?}{\emptyset}\quad\text{as required.}$$

\item[Case $r\in\mset{!,\id}$:] Similar to the case of \qt{?}.
\end{description}
\item[Cases where $P$ is $Q_1\plusop (Q_2\plusop Q_3)$, $Q_1\plusop Q_2$, $Q_1\parop Q_2$, $\idop{id}\zero$:] follow directly by the definition. 

\item[Case $P$ is $\idop{id}(Q_1\parop Q_2)$] We need to prove that  $\proj{\idop{id}(Q_1\parop Q_2)}{r}{\emptyset} \equiv \proj{\idop{id}Q_1}{r}{\emptyset}\parop\proj{\idop{id}Q_2}{r}{\emptyset}$ for any $r$. We have three cases depending on $r$.

\begin{description}
\item[Case $r=?$:] By definition of the projection function, we have that 
$$\proj{\idop{id}(Q_1\parop Q_2)}{?}{\emptyset}=\proj{Q_1\parop Q_2}{?}{\emptyset}=\proj{Q_1}{?}{\emptyset}\parop\proj{Q_2}{?}{\emptyset}$$
$$\proj{\idop{id}Q_1}{?}{\emptyset}\parop\proj{\idop{id}Q_2}{?}{\emptyset}=\proj{Q_1}{?}{\emptyset}\parop\proj{Q_2}{?}{\emptyset}\quad\text{as required.}$$

\item[Case $r=!$:] Similar to the case of \qt{?}.
\item[Case $r=\id$:] We have two cases: $r\neq\id$ or $r=\id$. The former case is similar to $\mset{?,!}$ while the latter case can be proved as follows;

$$\proj{\idop{id}(Q_1\parop Q_2)}{\id}{\emptyset}=\proj{Q_1\parop Q_2}{\id}{\emptyset}\parop\proj{Q_1\parop Q_2}{!}{\emptyset}=\proj{Q_1}{\id}{\emptyset}\parop\proj{Q_2}{\id}{\emptyset}\parop\proj{Q_1}{!}{\emptyset}\parop\proj{Q_2}{!}{\emptyset}$$
$$\proj{\idop{id}Q_1}{\id}{\emptyset}\parop\proj{\idop{\id}Q_2}{id}{\emptyset}=\proj{Q_1}{\id}{\emptyset}\parop\proj{Q_1}{!}{\emptyset}\parop \proj{Q_2}{\id}{\emptyset}\parop\proj{Q_2}{!}{\emptyset}$$
And we have that $\proj{Q_1}{\id}{\emptyset}\parop\proj{Q_2}{\id}{\emptyset}\parop\proj{Q_1}{!}{\emptyset}\parop\proj{Q_2}{!}{\emptyset}\equiv\proj{Q_1}{\id}{\emptyset}\parop\proj{Q_1}{!}{\emptyset}\parop \proj{Q_2}{\id}{\emptyset}\parop\proj{Q_2}{!}{\emptyset}$ as required.
\end{description}
\item[Case $P$ is $\idop{id_1}\idop{id_2}Q$] We need to prove that  $\proj{\idop{id_1}\idop{id_2}Q}{r}{\emptyset} \equiv \proj{\idop{id_2}\idop{id_1}Q}{r}{\emptyset}$ for any $r$. We have three cases depending on $r$.

\begin{description}
\item[Cases $r\in\mset{?,!}$:] By defintion $\proj{\idop{\id_1}\idop{\id_2}Q}{?}{\emptyset} \equiv \proj{\idop{\id_2}\idop{\id_1}Q}{?}{\emptyset}=\proj{\idop{\id_1}\idop{\id_2}Q}{!}{\emptyset} \equiv \proj{\idop{\id_2}\idop{\id_1}Q}{!}{\emptyset}=\proj{Q}{?}{\emptyset}=\proj{Q}{!}{\emptyset}$

\item[Case $r=\id$:] We have three cases:
\begin{description}
\item[1) $\id\neq\id_1\land\id\neq\id_2$:] By defintion we have that $\proj{\idop{\id_1}\idop{\id_2}Q}{\id}{\emptyset} \equiv \proj{\idop{\id_2}\idop{\id_1}Q}{\id}{\emptyset}=\proj{Q}{\id}{\emptyset}$.

\item [2) $\id=\id_1$:] By definition we have that: 
$$\proj{\idop{\id_1}\idop{\id_2}Q}{\id_1}{\emptyset}=\proj{\idop{\id_2}Q}{\id_1}{\emptyset}\parop\proj{\idop{\id_2}Q}{!}{\emptyset}=\proj{Q}{\id_1}{\emptyset}\parop\proj{Q}{!}{\emptyset}$$
$$\proj{\idop{\id_2}\idop{\id_1}Q}{\id_1}{\emptyset}=\proj{\idop{\id_1}Q}{\id_1}{\emptyset}=\proj{Q}{\id_1}{\emptyset}\parop\proj{Q}{!}{\emptyset}\quad\text{as required.}$$
\item [3) $\id=\id_2$:] Similar to case (2).

\end{description}
\end{description}

\item[Case $P$ is $\rec{Q}$] We need to prove that  $\proj{\rec{Q}}{r}{\emptyset} \equiv \proj{Q[\rec{Q}/X]}{r}{\emptyset}$ for any $r$. Directly by Lemma \ref{lem:recdel}.

\end{description}
\end{proof}

Lemma \ref{lem:outprojred} states basically that the projection function is invariant to reduction in the case of output and reactive projections. We will use this result mainly in the proof of the main theorem. 
\begin{lemma}[Preservation of $!$/$?$-Projections Under Reduction]\label{lem:outprojred}
If $\sys{\Delta}{P} \rTo{}{} \sys{\Delta'}{Q}$ then $\proj{P}{r}{\empty} \equiv \proj{Q}{r}{\empty}$ where $r\in\mset{!,?}$.
\end{lemma}
\begin{proof}
By induction on the length of the derivation $\sys{\Delta}{P} \rTo{}{} \sys{\Delta'}{Q}$.

\begin{description}
\item[Case 1:] Rule \rulename{Bin} is applied: We have that $\sys{\Delta}{\idop{\id}(\filter{o}{i}P\plusop S)}\red\\
\sys{\Delta'}{\filter{o}{i}(\idop{{\id'}}P)\plusop S}$ and we need to prove that $\proj{\idop{\id}(\filter{o}{i}P\plusop S)}{r}{\emptyset}\equiv\proj{\filter{o}{i}(\idop{{\id'}}P)\plusop S}{r}{\emptyset}$. We have different cases:

\begin{description}
\item[Case $r=?$:] By definition of the projection function, we have that 
$$\proj{\idop{\id}(\filter{o}{i}P\plusop S)}{?}{\emptyset}=\mconfig{i}\func{d}{}?.\proj{P}{!}{\emptyset}\parop\proj{P}{?}{\emptyset}\parop \proj{S}{?}{\emptyset}$$
$$\proj{\filter{o}{i}(\idop{{\id'}}P)\plusop S}{?}{\emptyset}=\mconfig{i}\func{d}{}?.\proj{\idop{id'}P}{!}{\emptyset}\parop\proj{\idop{id'}P}{?}{\emptyset}\parop \proj{S}{?}{\emptyset}$$
By definition again, we know that $\proj{\idop{id'}P}{!}{\emptyset}=\proj{P}{!}{\emptyset}$ and $\proj{\idop{id'}P}{?}{\emptyset}=\proj{P}{?}{\emptyset}$ and we have that $\proj{\idop{\id}(\filter{o}{i}P\plusop S)}{?}{\emptyset}\equiv\proj{\filter{o}{i}(\idop{{\id'}}P)\plusop S}{?}{\emptyset}$ as required.
\item[Case $r=!$:] By definition of the projection function, we have that 

$$\proj{\idop{\id}(\filter{o}{i}P\plusop S)}{!}{\emptyset}=\mconfig{o}\func{d}{}! \plusop \proj{S}{!}{\emptyset}$$
$$\proj{\filter{o}{i}(\idop{{\id'}}P)\plusop S}{!}{\emptyset}=\proj{\filter{o}{i}(\idop{{\id'}}P)}{!}{\emptyset}\plusop \proj{S}{!}{\emptyset}$$

By definition again, we know that $\proj{\filter{o}{i}(\idop{{\id'}}P)}{!}{\emptyset}=\mconfig{o}\func{d}{}!$ and we have that $\proj{\idop{\id}(\filter{o}{i}P\plusop S)}{!}{\emptyset}\equiv\proj{\filter{o}{i}(\idop{{\id'}}P)\plusop S}{!}{\emptyset}$ as required.

\end{description}

\item[Case 2] Rules \rulename{Brd} and \rulename{Loc} can be proved in a similar manner.

\item[Case 3] Rule {\rulename{Synch}}: We need to prove that if $\sys{\Delta}{\filter{o}{i}P}\red
\sys{\Delta'}{\filter{o}{i}P'}$ then $\proj{\filter{o}{i}P}{r}{\emptyset}\equiv
\proj{{\filter{o}{i}P'}}{r}{\emptyset}$. We have different cases:

\begin{description}
\item[Case $r=?$:] By definition of the projection function, we have that 
$$\proj{\filter{o}{i}P}{?}{\emptyset}=\mconfig{i}\func{d}{}?.\proj{P}{!}{\emptyset}\parop\proj{P}{?}{\emptyset}$$
$$\proj{\filter{o}{i}P')}{?}{\emptyset}=\mconfig{i}\func{d}{}?.\proj{P'}{!}{\emptyset}\parop\proj{P'}{?}{\emptyset}$$
By the induction hypothesis, we have that $\proj{P'}{!}{\emptyset}=\proj{P}{!}{\emptyset}$ and $\proj{P'}{?}{\emptyset}=\proj{P}{?}{\emptyset}$ and we have that $\proj{\filter{o}{i}P}{?}{\emptyset}\equiv
\proj{{\filter{o}{i}P'}}{?}{\emptyset}$ as required.
\item[Case $r=!$:] By definition of the projection function, we have that 

$$\proj{\filter{o}{i}P}{!}{\emptyset}=\proj{\filter{o}{i}P'}{!}{\emptyset}=\mconfig{o}\func{d}{}!$$

and we have that $\proj{\filter{o}{i}P}{!}{\emptyset}\equiv
\proj{{\filter{o}{i}P'}}{!}{\emptyset}$ as required.

\end{description}

\item[Case 4] Rule \rulename{Id}: We need to prove that if $\sys{\Delta}{\idop{\id}P}\red
\sys{\Delta'}{\idop{\id}P'}
$ then $\proj{{\idop{\id}P}}{r}{\emptyset}\equiv\proj{{\idop{\id}P'}}{r}{\emptyset}$. We have different cases:

\begin{description}
\item[Case $r=?$:] By definition of the projection function, we have that 
$$\proj{\idop{id}P}{?}{\emptyset}=\proj{P}{?}{\emptyset}\quad \text{and}\quad\proj{\idop{id}P')}{?}{\emptyset}=\proj{P'}{?}{\emptyset}$$
By the induction hypothesis, we have that $\proj{P'}{?}{\emptyset}=\proj{P}{?}{\emptyset}$ and we have that $\proj{{\idop{\id}P}}{?}{\emptyset}\equiv\proj{{\idop{\id}P'}}{?}{\emptyset}$ as required.

\item[Case $r=!$:] By definition of the projection function, we have that 
$$\proj{\idop{id}P}{!}{\emptyset}=\proj{P}{!}{\emptyset}\quad \text{and}\quad\proj{\idop{id}P')}{!}{\emptyset}=\proj{P'}{!}{\emptyset}$$
By the induction hypothesis, we have that $\proj{P'}{!}{\emptyset}=\proj{P}{!}{\emptyset}$ and we have that $\proj{{\idop{\id}P}}{!}{\emptyset}\equiv\proj{{\idop{\id}P'}}{!}{\emptyset}$ as required.

\end{description}

\item[Case 5] Rule \rulename{Sum}: We need to prove that if $\sys{\Delta}{P_1\plusop P_2}\red
\sys{\Delta'}{P'_1\plusop P_2}
$ then $\proj{{P_1\plusop P_2}}{r}{\emptyset}\equiv\proj{{P'_1\plusop P_2}}{r}{\emptyset}$. We have different cases:

\begin{description}
\item[Case $r=?$:] By definition of the projection function, we have that 
$$\proj{P_1\plusop P_2}{?}{\emptyset}=\proj{P_1}{?}{\emptyset}\parop\proj{P_2}{?}{\emptyset}$$
$$\proj{P'_1\plusop P_2}{?}{\emptyset}=\proj{P'_1}{?}{\emptyset}\parop\proj{P_2}{?}{\emptyset}$$
By the induction hypothesis, we have that $\proj{P'_1}{?}{\emptyset}=\proj{P_1}{?}{\emptyset}$ and we have that $\proj{{P_1\plusop P_2}}{?}{\emptyset}\equiv\proj{{P'_1\plusop P_2}}{?}{\emptyset}$ as required.
\item[Case $r=!$:] By definition of the projection function, we have that 
$$\proj{P_1\plusop P_2}{!}{\emptyset}=\proj{P_1}{!}{\emptyset}\plusop\proj{P_2}{!}{\emptyset}$$
$$\proj{P'_1\plusop P_2}{!}{\emptyset}=\proj{P'_1}{!}{\emptyset}\plusop\proj{P_2}{!}{\emptyset}$$
By the induction hypothesis, we have that $\proj{P'_1}{!}{\emptyset}=\proj{P_1}{!}{\emptyset}$ and we have that $\proj{{P_1\plusop P_2}}{!}{\emptyset}\equiv\proj{{P'_1\plusop P_2}}{!}{\emptyset}$ as required.

\end{description}

\item[Case 6] Rule \rulename{Par}: We need to prove that if $\sys{\Delta}{P_1\parop P_2}\red
\sys{\Delta'}{P'_1\parop P_2}
$ then $\proj{{P_1\parop P_2}}{r}{\emptyset}\equiv\proj{{P'_1\parop P_2}}{r}{\emptyset}$. This case can be proved in a similar manner to the previous case but by using rule \rulename{Par} instead of rule \rulename{Sum}.

\item[Case 7] Rule \rulename{Struct}: We need to prove that if $\sys{\Delta}{P}\red
\sys{\Delta'}{Q}
$ then $\proj{{P}}{r}{\emptyset}\equiv\proj{{Q}}{r}{\emptyset}$.

From rule \rulename{Struct}, we have that $\sys{\Delta}{P}\red
\sys{\Delta'}{Q}
$ if $\sys{\Delta}{P'}\red
\sys{\Delta'}{Q'}
$ where $P'\equiv P$ and $Q'\equiv Q$. By the induction hypothesis we have that  $\proj{{P'}}{r}{\emptyset}\equiv\proj{{Q'}}{r}{\emptyset}$. But $P'\equiv P$ and $Q'\equiv Q$, so we apply Lemma \ref{lem:structP} and conclude the proof.
\end{description}

\end{proof}

The following Lemma ensures that the parallel composition is merely interleaving and does not influence the behaviour of any of its sub-definitions, namely if a definition is able to take a step when isolated then this step will be possible also when put in parallel with any other definition. 
\begin{lemma}[Closure of LTS Under Definition Context]\label{lem:cong}
If $\node{s}{D_1}\rTo{\hat{\lambda}}{}\node{s'}{D_2}$ then $\node{s}{D_1\parop D_3}\rTo{\hat{\lambda}}{}\node{s'}{D_2\parop D_3}$ for any $D_3$.
\end{lemma}
\begin{proof}
By induction on the length of the derivation  $\node{s}{D_1\parop D_3}\rTo{\hat{\lambda}}{}\node{s'}{D_2\parop D_3}$.
\end{proof}

Lemma \ref{lem:struct} ensures that structurally equivalent protocols have structurally equivalent projections under any possible network configuration $\Delta$. 

\begin{lemma}[Preservation of Structural Congruence Under Network Projection]\label{lem:struct}
If $P \equiv Q$ then $\tproj{\sys{\Delta}{P}} \equiv \tproj{\sys{\Delta}{Q}}$ for any $\Delta$.
\end{lemma}
\begin{proof}
The proof proceeds by relying on the definition of the projection function and Lemma \ref{lem:structP}. By definition we have that: 

$$\tproj{\sys{\Delta}{P}}  \deff  \Pi_{\forall \id \in \mathit{dom}(\Delta)} 
(\node{\Delta(\id)}{
\proj{P}{?}{\emptyset} \parop \proj{P}{\id}{\emptyset}})$$ 

$$\tproj{\sys{\Delta}{Q}}  \deff  \Pi_{\forall \id \in \mathit{dom}(\Delta)} 
(\node{\Delta(\id)}{
\proj{Q}{?}{\emptyset} \parop \proj{Q}{\id}{\emptyset}})$$

Since $P\equiv Q$, we have that, by Lemma \ref{lem:structP} and regardless of $\Delta$, $\proj{P}{?}{\emptyset}\equiv\proj{Q}{?}{\emptyset}$ and $\proj{P}{id}{\emptyset}\equiv\proj{Q}{id}{\emptyset}$. Thus $\tproj{\sys{\Delta}{P}} \equiv \tproj{\sys{\Delta}{Q}}$ as required. 

\end{proof}

The soundness of the operational correspondence (i.e., the first statement of Theorem 1) is proved in the following lemma.
\begin{lemma}[Soundness]\label{lem:sound}
If $\sys{\Delta}{P} \rTo{}{} \sys{\Delta'}{Q}$ then $\tproj{\sys{\Delta}{P}} \rTo{\hat{\lambda}}{\!\!_\equiv\;} \tproj{\sys{\Delta'}{Q}}$.
\end{lemma}

\begin{proof}   
The proof proceeds by induction on the length of the derivation $\sys{\Delta}{P} \rTo{}{} \sys{\Delta}{Q}$.
\begin{description}
\item[Case 1] Rule \rulename{Bin} is applied: We have that $\sys{\Delta}{\idop{\id}(\filter{o}{i}P\plusop S)}\red\\
\sys{\Delta'}{\filter{o}{i}(\idop{{\id'}}P)\plusop S}$ and we need to prove that $\tproj{\sys{\Delta}{\idop{\id}(\filter{o}{i}P\plusop S)}}\rTo{\hat{\lambda}}{\!\!_\equiv\;}\tproj{\sys{\Delta'}{\filter{o}{i}(\idop{{\id'}}P)\plusop S}}$.

Since a binary interaction happened in the global model, we know that there must be a sender $\id$ with $\Delta(\id)=\state_1$ and a receiver $\id'$ with $\Delta(\id')=\state_2$ such that for some $d\in\mset{\toup,\toright}$ we have that $d(\Delta,\id)=id'$, $\state_1\models o$ and $\state_2\models i$. As a result of synchronisation on $f$, we have also that $\state'_1= \func{d}{}!(\state_1, \id')$ and $\state'_2= \func{d}{}?(\state_2,\id)$. This can be concluded from the definition of  $\upd(\id,\id',\func{d}{},\Delta)$ and thus  
$\Delta'=\Delta[ \id \mapsto \func{d}{}!(\state_1, \id'), \id' \mapsto \func{d}{}?(\state_2,\id)]$.

From the definition of the main projection rule in Table~\ref{tab:proj}, we have that:
$$\tproj{\sys{\Delta}{Q}}=\net\ \|\ \node{s_1}{\proj{Q}{?}{\emptyset} \parop \proj{Q}{id}{\emptyset}}\ \|\ \node{s_2}{\proj{Q}{?}{\emptyset} \parop \proj{Q}{id}{\emptyset}}$$

Where $Q=\idop{\id}(\filter{o}{i}P\plusop S)$ and $\net$ is the rest of the nodes in the network. We do not expand $\net$ because it does not contribute to the transition. By Table~\ref{tab:proj}, we have that 

$\node{s_1}{\proj{Q}{?}{\emptyset} \parop \proj{Q}{id}{\emptyset}} \deff  \node{\state_1}{
\mconfig{i}\func{d}{}?.\proj{P}{!}{\emptyset}\parop\proj{P}{?}{\emptyset}\parop \proj{S}{?}{\emptyset}\parop \parop {\proj{{P}}{\id}{\emptyset}\parop \proj{{S}}{\id}{\emptyset} \parop( \mconfig{o}\func{d}{}!} \plusop \proj{{S}}{!}{\emptyset})}$

$\node{s_2}{\proj{Q}{?}{\emptyset} \parop \proj{Q}{id}{\emptyset}}  \deff  \node{\state_2}{
\mconfig{i}\func{d}{}?.\proj{P}{!}{\emptyset}\parop\proj{P}{?}{\emptyset}\parop \proj{S}{?}{\emptyset}\parop {\proj{{P}}{\id'}{\emptyset}\parop \proj{{S}}{\id'}{\emptyset}}}$

The overall network evolves by rule \rulename{Com} where $\node{s_1}{D_1}$ applies either rule \rulename{oBinU} or rule \rulename{oBinR} and $\node{s_2}{D_2}$ applies rule \rulename{iBin}, we have that:

$$\net\ \|\ \node{s_1}{\proj{Q}{?}{\emptyset} \parop \proj{Q}{id}{\emptyset}}\ \|\ \node{s_2}{\proj{Q}{?}{\emptyset} \parop \proj{Q}{id'}{\emptyset}}\rTo{\tau}{}\net\ \|\ \node{\state_1}{D'_1}\ \|\ \node{\state_2}{D'_2}$$

\[
D_1'= {
\mconfig{i}\func{d}{}?.\proj{P}{!}{\emptyset}\parop\proj{P}{?}{\emptyset}\parop \proj{S}{?}{\emptyset}\parop \parop {\proj{{P}}{\id}{\emptyset}\parop \proj{{S}}{\id}{\emptyset} }}
\]

\[
D_2'= {
\mconfig{i}\func{d}{}?.\proj{P}{!}{\emptyset}\parop\proj{P}{!}{\emptyset}\parop\proj{P}{?}{\emptyset}\parop \proj{S}{?}{\emptyset}\parop {\proj{{P}}{\id'}{\emptyset}\parop \proj{{S}}{\id'}{\emptyset}}}
\]

Now, we need to show that $\tproj{\sys{\Delta'}{\filter{o}{i}(\idop{{\id'}}P)\plusop S}}\equiv\net\ \|\ \node{\state_1}{D'_1}\ \|\ \node{\state_2}{D'_2}$. By Table~\ref{tab:proj}, we have that 

$$\tproj{\sys{\Delta'}{\overbrace{\filter{o}{i}(\idop{{\id'}}P)\plusop S}^{Q'}}}=\net\ \|\ \node{s'_1}{\proj{Q'}{?}{\emptyset} \parop \proj{Q'}{id}{\emptyset}}\ \|\ \node{s'_2}{\proj{Q'}{?}{\emptyset} \parop \proj{Q'}{id}{\emptyset}}$$

By applying the projection function, we have that $D'_1\equiv \proj{Q'}{?}{\emptyset} \parop \proj{Q'}{id}{\emptyset}$ and $D'_2\equiv \proj{Q'}{?}{\emptyset} \parop \proj{Q'}{id'}{\emptyset}$ as required.

\item[Case 2] Rules \rulename{Brd} and \rulename{Loc} can be proved in a similar manner.

\item[Case 3] Rule \rulename{Synch}: We need to prove that if $\sys{\Delta}{\filter{o}{i}P}\red
\sys{\Delta'}{\filter{o}{i}P'}$ then $\tproj{\sys{\Delta}{\filter{o}{i}P}}\rTo{\hat{\lambda}}{\!\!_\equiv\;}
\tproj{\sys{\Delta'}{\filter{o}{i}P'}}$. From rule \rulename{Synch}, we know that  $\sys{\Delta}{\filter{o}{i}P}\red
\sys{\Delta'}{\filter{o}{i}P'}$ if $\sys{\Delta}{P}\red
\sys{\Delta'}{P'}$. By the induction hypothesis, we have that $\tproj{\sys{\Delta}{P}}\rTo{\hat{\lambda}}{\!\!_\equiv\;}
\tproj{\sys{\Delta'}{P'}}$. By relying on the definition of the projection function, we have that:

$$\overbrace{\Pi_{\forall \id \in \mathit{dom}(\Delta)} 
(\node{\Delta(\id)}{
\proj{P}{?}{\emptyset} \parop \proj{P}{\id}{\emptyset}})}^{\tproj{\sys{\Delta}{P}}}\rTo{\hat{\lambda}}{\!\!_\equiv\;}\overbrace{\Pi_{\forall \id \in \mathit{dom}(\Delta')} 
(\node{\Delta'(\id)}{
\proj{P'}{?}{\emptyset} \parop \proj{P'}{\id}{\emptyset}})}^{\tproj{\sys{\Delta'}{P'}}}$$ We also know by definition that $$\tproj{\sys{\Delta}{\filter{o}{i}P}}=\Pi_{\forall \id \in \mathit{dom}(\Delta)} 
(\node{\Delta(\id)}{
\mconfig{i}\func{d}{}?.\proj{P}{!}{\emptyset}\parop\proj{P}{?}{\emptyset} \parop \proj{P}{\id}{\emptyset}})$$ and $$\tproj{\sys{\Delta'}{\filter{o}{i}P'}}=\Pi_{\forall \id \in \mathit{dom}(\Delta')} 
(\node{\Delta'(\id)}{
\mconfig{i}\func{d}{}?.\proj{P'}{!}{\emptyset}\parop\proj{P'}{?}{\emptyset}  \parop \proj{P'}{\id}{\emptyset}})$$

We can prove that $\tproj{\sys{\Delta}{\filter{o}{i}P}}\rTo{\hat{\lambda}}{\!\!_\equiv\;}
\tproj{\sys{\Delta'}{\filter{o}{i}P'}}$ by the induction hypothesis, by Lemma~\ref{lem:outprojred}, we have that $\proj{P}{!}{\emptyset}\equiv \proj{P'}{!}{\emptyset}$ and by Lemma~\ref{lem:cong} and finally by applying rule \rulename{Com} or rule \rulename{Par} depending on $\hat{\lambda}$ we conclude the proof.

%
%

\item[Case 4] Rule \rulename{Id}: We need to prove that if $\sys{\Delta}{\idop{\id}P}\red
\sys{\Delta'}{\idop{\id}P'}
$ then $\tproj{\sys{\Delta}{\idop{\id}P}}\rTo{\hat{\lambda}}{\!\!_\equiv\;}\tproj{\sys{\Delta'}{\idop{\id}P'}}$. From rule \rulename{Id}, we know that  $\sys{\Delta}{\idop{id}P}\red
\sys{\Delta'}{\idop{id}P'}$ if $\sys{\Delta}{P}\red
\sys{\Delta'}{P'}$. By the induction hypothesis, we have that $\tproj{\sys{\Delta}{P}}\rTo{\hat{\lambda}}{\!\!_\equiv\;}
\tproj{\sys{\Delta'}{P'}}$. 

By the definition of the projection function, we can rewrite the projection with respect to a single node $\id'$ with $\Delta(\id')=\state_1$ as follows: $\tproj{\sys{\Delta}{\idop{id}P}}=\net\ \|\ \node{\Delta(\id')}{
\proj{\idop{id}P}{?}{\emptyset} \parop \proj{\idop{id}P}{\id'}{\emptyset}}$ where $\net$ is the rest of the nodes. By the induction hypothesis we have that:

$$
\net\ \|\ \node{\Delta(\id')}{
 \proj{P}{?}{\emptyset} \parop \proj{P}{\id'}{\emptyset}}\rTo{\hat{\lambda}}{} \net'\ \|\
\node{\Delta'(\id')}{
\proj{P'}{?}{\emptyset} \parop \proj{P'}{\id'}{\emptyset}}$$ 

We have two cases for the projection of $\tproj{\sys{\Delta}{\idop{id}P}}$: The case when $\id'\neq\id$ and the other one when $\id'=\id$. The projection according to the former case proceeds as follows: 

$$\tproj{\sys{\Delta}{\idop{id}P}}=\net\ \|\ 
\node{\Delta(\id')}{\proj{P}{?}{\emptyset} \parop \proj{P}{\id'}{\emptyset}}$$ 
and 
$$\tproj{\sys{\Delta'}{\idop{id}P'}}=\net'\ \|\
\node{\Delta'(\id')}{\proj{P'}{?}{\emptyset}  \parop \proj{P'}{\id'}{\emptyset}}$$

Clearly, this case follows directly by the induction hypothesis. For the latter case, $id'= id$, the projection proceeds as follows:

$$\tproj{\sys{\Delta}{\idop{id}P}}=\net\ \|\ 
\node{\Delta(\id')}{\proj{P}{?}{\emptyset} \parop \proj{P}{!}{\emptyset} \parop \proj{P}{\id'}{\emptyset}}$$ 
and 
$$\tproj{\sys{\Delta'}{\idop{id}P'}}=\net'\ \|\
\node{\Delta'(\id')}{\proj{P'}{?}{\emptyset} \parop \proj{P'}{!}{\emptyset} \parop \proj{P'}{\id'}{\emptyset}}$$

We have that $\tproj{\sys{\Delta}{\idop{id}P}}\rTo{\hat{\lambda}}{\!\!_\equiv\;}
\tproj{\sys{\Delta'}{\idop{id}P'}}$ can be proved by the induction hypothesis, by Lemma~\ref{lem:outprojred}, we have that $\proj{P}{!}{\emptyset}\equiv \proj{P'}{!}{\emptyset}$, by applying Lemma~\ref{lem:cong} and finally by applying rule \rulename{Com} or rule \rulename{Par} depending on $\hat{\lambda}$ we conclude the proof.

\item[Case 5] Rule \rulename{Sum}: We need to prove that if $\sys{\Delta}{P_1\plusop P_2}\red
\sys{\Delta'}{P'_1\plusop P_2}
$ then $\tproj{\sys{\Delta}{P_1\plusop P_2}}\rTo{\hat{\lambda}}{\!\!_\equiv\;}\tproj{\sys{\Delta'}{P'_1\plusop P_2}}$. From rule \rulename{Sum}, we know that  $\sys{\Delta}{P_1\plusop P_2}\red
\sys{\Delta'}{P'_1\plusop P_2}$ if $\sys{\Delta}{P_1}\red
\sys{\Delta'}{P'_1}$. By the induction hypothesis, we have that $\tproj{\sys{\Delta}{P_1}}\rTo{\hat{\lambda}}{\!\!_\equiv\;}
\tproj{\sys{\Delta'}{P'_1}}$. 

By the definition of the projection function, we can rewrite the projection with respect to a single node $\id'$ with $\Delta(\id')=\state_1$ as follows: $\tproj{\sys{\Delta}{P_1\plusop P_2}}=\net\ \|\ \node{\Delta(\id')}{
\proj{P_1\plusop P_2}{?}{\emptyset} \parop \proj{P_1\plusop P_2}{\id'}{\emptyset}}$ where $\net$ is the rest of the nodes. By the induction hypothesis we have that:

$$
\net\ \|\ \node{\Delta(\id')}{
 \proj{P_1}{?}{\emptyset} \parop \proj{P_1}{\id'}{\emptyset}}\rTo{\hat{\lambda}}{} \net'\ \|\
\node{\Delta'(\id')}{
\proj{P'_1}{?}{\emptyset} \parop \proj{P'_1}{\id'}{\emptyset}}$$ 

The projection of $\tproj{\sys{\Delta}{P_1\plusop P_2}}$ proceeds as follows: 

$$\tproj{\sys{\Delta}{P_1\plusop P_2}}=\net\ \|\ 
\node{\Delta(\id')}{\proj{P_1}{?}{\emptyset} \parop \proj{P_1}{\id'}{\emptyset}\parop \proj{P_2}{?}{\emptyset} \parop \proj{P_2}{\id'}{\emptyset}}$$ 
and 
$$\tproj{\sys{\Delta'}{P'_1\plusop P_2}}=\net'\ \|\
\node{\Delta'(\id')}{\proj{P'_1}{?}{\emptyset}  \parop \proj{P'_1}{\id'}{\emptyset}\parop \proj{P_2}{?}{\emptyset} \parop \proj{P_2}{\id'}{\emptyset}}$$

We have that $\tproj{\sys{\Delta}{P_1\plusop P_2}}\rTo{\hat{\lambda}}{\!\!_\equiv\;}
\tproj{\sys{\Delta'}{P'_1\plusop P_2}}$ can be proved by the induction hypothesis, by applying Lemma~\ref{lem:cong} and finally by applying rule \rulename{Com} or rule \rulename{Par} depending on $\hat{\lambda}$ we conclude the proof.

\item[Case 6] Rule \rulename{Par}: We need to prove that if $\sys{\Delta}{P_1\parop P_2}\red
\sys{\Delta'}{P'_1\parop P_2}
$ then $\tproj{\sys{\Delta}{P_1\parop P_2}}\rTo{\hat{\lambda}}{\!\!_\equiv\;}\tproj{\sys{\Delta'}{P'_1\parop P_2}}$. This case can be proved in a similar manner to the previous case but by using rule \rulename{Par} instead of rule \rulename{Sum}.

\item[Case 7] Rule \rulename{Struct} We need to prove that if $\sys{\Delta}{P}\red
\sys{\Delta'}{Q}
$ then $\tproj{\sys{\Delta}{P}}\rTo{\hat{\lambda}}{\!\!_\equiv\;}\tproj{\sys{\Delta'}{Q}}$.

From rule \rulename{Struct}, we have that $\sys{\Delta}{P}\red
\sys{\Delta'}{Q}
$ if $\sys{\Delta}{P'}\red
\sys{\Delta'}{Q'}
$ where $P'\equiv P$ and $Q'\equiv Q$. By relying on Lemma \ref{lem:struct}, we have that $\tproj{\sys{\Delta}{P}}\equiv\tproj{\sys{\Delta}{P'}}$ and $\tproj{\sys{\Delta}{Q}}\equiv\tproj{\sys{\Delta}{Q'}}$ and by the induction hypothesis we have that $\tproj{\sys{\Delta}{P}}\rTo{\hat{\lambda}}{\!\!_\equiv\;}\tproj{\sys{\Delta'}{Q}}$ as required.
\end{description}
\end{proof}

\subsection{Operational Correspondence - Completeness}
\label{sec:completeness}

The proofs presented in the remainder make use of the principle captured by Lemma~\ref{lem:structP},
which attests that the projection of structurally equivalent protocols yields structurally equivalent definitions.

We first characterise definitions in general (Proposition~\ref{pro:normalform}) and then the ones resulting from the different 
types of projection (Lemma~\ref{lem:projnormalform}), used in
the completeness proof (Lemma~\ref{lem:completeness}) to help characterise the definitions in the interacting nodes. 
We then state a crucial auxiliary result for the completeness proof that characterises the projection of the active nodes
(Lemma~\ref{lem:projactive}),
since reduction in configurations is given in terms of changes in active nodes. We also ensure that synchronisation
actions have unique counterparts in the respective projection (Lemma~\ref{lem:filterproj} and Lemma~\ref{lem:inpuni})
and that output (choices) occurring in the result of a projection have a synchronisation action (summation) counterpart
in the source protocol (Lemma~\ref{lem:projinv}). Such relations allow us to precisely identify the synchronisation action in the protocol related
to the interacting input/output definitions in the network. We then present the results that allow to characterise networks
and definitions based on observable behaviour (Lemma~\ref{lem:invnettau} - Lemma~\ref{lem:invbroout}), used in the
proof to characterise the interacting nodes. Finally,
we state that reduction is closed in (active) language contexts (Lemma~\ref{lem:redclosurecntxt}).

We start by informing on the structure of definitions in general, namely that any choice
can be described as a summation of outputs, that reactions are a parallel composition of choices,
and that definitions are a parallel composition of input prefixed reactions in parallel with a reaction.
In the following we use $\Pi_{i \in I} D_i$ to abbreviate $D_1 \parop \ldots \parop D_k$, 
and $\Sigma_{i \in I} D_i$ to abbreviate $D_1 \plusop \ldots \plusop D_k$,
when $I = 1,\ldots,k$ (if $I = \emptyset$ then $\Pi_{i \in I} D_i$ denotes $\zero$). 

\begin{proposition}[Normal Form]
\label{pro:normalform}
For any choice $C$, reaction $R$, and definition $D$ we have that 
$$C \equiv \Sigma_{l \in L} \mconfig{c_l}\fun_{l}^{d_l}!$$
$$R \equiv \Pi_{j \in J} C_j$$
$$D \equiv (\Pi_{i \in I} \mconfig{c_i}\fun_{i}^{d_i}?.{R_i}) \parop R'$$
\end{proposition}
\begin{proof}
By induction of the structure of $C$, $R$, and $D$, following expected lines. 
In case $C$ is $\mconfig{c}\fun^{d}!$ the result is direct, and in case $C$ is $C_1 \plusop C_2$
the result follows by gathering the summations obtained by induction hypothesis. In case
$R$ is $C$ or $\zero$ the result is direct, and in case $R$ is $R_1\parop R_2$ the result follows
 by gathering the products obtained by induction hypothesis. In case $D$ is
$\mconfig{c}\fun^{d}?.{R}$ or $R$ the result is direct (since $D \parop \zero \equiv D$),
and in case $D$ is $D_1 \parop D_2$ the result follows by gathering the products and
reactions obtained by induction hypothesis.
\end{proof}

The following result characterises the result of the different types of projection, namely
that $!$-projection and $\id$-projections yield reactions and that $?$-projection yields 
a parallel composition of input prefixed reactions. We also identify an auxiliary property
that ensures the $!$-projection of recursion guarded processes yields a reaction regardless
of the mapping considered in the projection.

\begin{lemma}[Projection Normal Form]
\label{lem:projnormalform}
Let $P$ and $\sigma$ be such that for all $\varX \in \mathit{fv}(P)$ and any $\sigma''$
it is the case that $\proj{\sigma(\varX)}{!}{\sigma''} \equiv R'''$. We have that
\begin{enumerate}
\item If recursion is guarded in $P$ then $\proj{P}{!}{\sigma'} \equiv R$ for any $\sigma'$.
\item $\proj{P}{!}{\sigma} \equiv R'$.
\item $\proj{P}{?}{\sigma} \equiv \Pi_{i \in I} \mconfig{c_i}\fun_{i}^{d_i}?.{R_i}$.
\item $\proj{P}{\id}{\sigma} \equiv R''$.
\end{enumerate}
\end{lemma}
\begin{proof}
By induction on the structure of P, where each item relies on previous ones (except for \emph{3.} and \emph{4.} that are independent).

\noindent{\bf (Case $P$ is $\filter{o}{i}Q$)}

{\emph{1.}}
Direct since $\proj{\filter{o}{i}Q}{!}{\sigma'}$ is defined as $\mconfig{o}\fun^{d}!$.

{\emph{2.}}
Likewise.

{\emph{3.}}
We have that $\proj{\filter{o}{i}Q}{?}{\sigma}$ by definition 
is $\mconfig{i}\fun^{d}?.{\proj{Q}{!}{\sigma}} \parop \proj{Q}{?}{\sigma} $. By induction
hypothesis we have that $\proj{Q}{?}{\sigma} \equiv \Pi_{i \in I} \mconfig{c_i}\fun_{i}^{d_i}?.{R_i}$
and by \emph{2.} we have that ${\proj{Q}{!}{\sigma}}\equiv R'$, hence we conclude
$\proj{\filter{o}{i}Q}{?}{\sigma} \equiv \mconfig{i}\fun^{d}?.R' \parop \Pi_{i \in I} \mconfig{c_i}\fun_{i}^{d_i}?.{R_i}$.

\emph{4.}
We have that $\proj{\filter{o}{i}Q}{\id}{\sigma}$ by definition is $\proj{Q}{\id}{\sigma}$
hence the result follows directly from the induction hypothesis.

\noindent{\bf(Case $P$ is $\idop{\id'}Q$)}

{\emph{1.}}
We have that $\proj{\idop{\id'}Q}{!}{\sigma'}$ by definition is $\proj{Q}{!}{\sigma'}$ hence the result follows directly from the induction hypothesis.

{\emph{2-3.}}
Likewise.

{\emph{4.}}
We have that $\proj{\idop{\id'}Q}{\id}{\sigma}$ by definition is $\proj{Q}{\id}{\sigma}$ if $\id \neq \id'$ in which case the result follows directly from the induction hypothesis. In case $\id = \id'$ we have that $\proj{\idop{\id'}Q}{\id}{\sigma}$ is defined as $\proj{Q}{\id}{\sigma} \parop 
\proj{Q}{!}{\sigma}$. By \emph{2.} we have that $\proj{Q}{!}{\sigma} \equiv R'$ and by induction hypothesis we have that
$\proj{Q}{\id}{\sigma} \equiv R''$ hence $\proj{\idop{\id'}Q}{\id}{\sigma} \equiv R'' \parop R'$.

\noindent{\bf(Case $P$ is $\zero$)}

{\emph{1.}}
Direct since $\proj{\zero}{!}{\sigma'}$ is defined as $\zero$.

{\emph{2-4.}}
Likewise.

\noindent{\bf(Case $P$ is $\varX$)}

{\emph{1.}}
Does not apply.

{\emph{2.}}
We have that $\proj{\varX}{!}{\sigma}$ by definition is $\proj{\sigma(\varX)}{!}{\sigma}$ and by hypothesis we have that 
$\proj{\sigma(\varX)}{!}{\sigma''} \equiv R'''$ for any $\sigma''$.

{\emph{3.}}
Direct since $\proj{\varX}{?}{\sigma}$ by definition is $\zero$.

{\emph{4.}}
Likewise.

\noindent{\bf(Case $P$ is $\rec{Q}$)}

{\emph{1.}}
We have that $\proj{\rec{Q}}{!}{\sigma'}$ is defined as $\proj{Q}{!}{\sigma'[\varX \mapsto Q]}$. By induction hypothesis
we have that $\proj{Q}{!}{\sigma''} \equiv R$ for any $\sigma''$, hence $\proj{\rec{Q}}{!}{\sigma'} \equiv R$ for any $\rho'$.

{\emph{2.}}
We have that $\proj{\rec{Q}}{!}{\sigma}$ is defined as $\proj{Q}{!}{\sigma[\varX \mapsto Q]}$. By \emph{1.} we conclude
that $\proj{Q}{!}{\sigma''} \equiv R$ for any $\sigma''$. Then, by induction hypothesis 
we have that $\proj{Q}{!}{\sigma[\varX \mapsto Q]} \equiv R'$ and hence $\proj{\rec{Q}}{!}{\sigma} \equiv R'$. 

{\emph{3-4.}}
Likewise.

\noindent{\bf(Case $P$ is $Q_1 \parop Q_2$)}

{\emph{1.}}
We have that $\proj{Q_1 \parop Q_2}{!}{\sigma'}$ by definition is $\proj{Q_1}{!}{\sigma'} \parop \proj{Q_2}{!}{\sigma'}$
hence the result follows from the induction hypothesis.

{\emph{2-4.}}
Likewise.

\noindent{\bf(Case $P$ is $S_1 \plusop S_2$)}

{\emph{1.}}
We have that $\proj{S_1 \plusop S_2}{!}{\sigma'}$ by definition is $\proj{S_1}{!}{\sigma'} \plusop \proj{S_2}{!}{\sigma'}$
hence the result follows from the induction hypothesis.

{\emph{2.}}
Likewise.

{\emph{3.}}
We have that $\proj{S_1 \plusop S_2}{?}{\sigma}$ by definition is $\proj{S_1}{?}{\sigma} \parop \proj{S_2}{?}{\sigma}$
hence the result follows from the induction hypothesis.

{\emph{4.}}
Likewise.
\end{proof}

\newcommand{\cntxt}{{\mathcal{C}}}

In order to identify the protocol language contexts where reduction may originate we introduce active contexts,
which include all language constructs except from recursion.

\begin{definition}[Active Contexts]
We denote by $\cntxt[\cdot]$ a global language term with one hole defined as follows
$$
\cntxt[\cdot] \;::= \; P \parop \cntxt[\cdot]  \quad |\quad \idop{\id}\cntxt[\cdot] \quad |\quad \filter{o}{i}\cntxt[\cdot] \quad | \quad S \plusop \cntxt[\cdot]
\quad | \quad \cdot
$$
\end{definition}


The following result relates the occurrence of an active node construct in an active context with the respective $\id$-projection, 
and states that active node constructs are ignored by other projections. Regarding the former, 
intuitively, we have that for each active node construct in the protocol the $\id$-projection yields the 
reaction obtained by the respective $!$-projection in parallel with the projection of the rest of the protocol.

\begin{lemma}[Active Node Projection]
\label{lem:projactive}
If $P \equiv \cntxt[\idop{\id}Q]$ then
 $\proj{P}{\id}{\emptyset} \equiv \proj{\cntxt[Q]}{\id}{\emptyset} \parop \proj{Q}{!}{\emptyset}$
and $\proj{P}{\pid}{\emptyset} \equiv \proj{\cntxt[Q]}{\pid}{\emptyset}$ when $\pid \neq \id$.
\end{lemma}
\begin{proof}
By induction on the structure of $\cntxt[\cdot]$.

(Case $\cntxt[\cdot]$ is $\cdot$) 

We have that $P \equiv \idop{\id}Q$ and by definition $\proj{\idop{\id}Q}{\id}{\emptyset}$
is $\proj{Q}{\id}{\emptyset} \parop \proj{Q}{!}{\emptyset}$ and $\proj{\idop{\id}Q}{\pid}{\emptyset}$ by definition is
$\proj{Q}{\pid}{\emptyset}$ when $\pid \neq \id$.

(Case $\cntxt[\cdot]$ is $P' \parop \cntxt'[\cdot]$) 

We have that $P \equiv P' \parop \cntxt'[\idop{\id}Q]$ and by definition
(\emph{i}) $\proj{P' \parop \cntxt'[\idop{\id}Q]}{\pid}{\emptyset}$ is 
$ \proj{P'}{\pid}{\emptyset} \parop \proj{\cntxt'[\idop{\id}Q]}{\pid}{\emptyset}$,
both when $\pid = \id$ and otherwise.
By induction hypothesis we have that (\emph{ii})
$\proj{\cntxt'[\idop{\id}Q]}{\id}{\emptyset} \equiv \proj{\cntxt'[Q]}{\id}{\emptyset} \parop \proj{Q}{!}{\emptyset}$
and that (\emph{iii}) $\proj{\cntxt'[\idop{\id}Q]}{\pid}{\emptyset} \equiv \proj{\cntxt'[Q]}{\pid}{\emptyset}$ for $\pid \neq \id$.
From (\emph{i}) and (\emph{ii}) and also by definition of projection we conclude $\proj{P' \parop \cntxt'[\idop{\id}Q]}{\id}{\emptyset} \equiv 
 \proj{P'}{\id}{\emptyset} \parop \proj{\cntxt'[Q]}{\id}{\emptyset} \parop \proj{Q}{!}{\emptyset}
 \equiv \proj{P' \parop \cntxt'[Q]}{\id}{\emptyset} \parop \proj{Q}{!}{\emptyset}$, and from (\emph{i}) and (\emph{iii}) and also
 by definition of projection we conclude 
 $\proj{P' \parop \cntxt'[\idop{\id}Q]}{\pid}{\emptyset} \equiv 
 \proj{P'}{\pid}{\emptyset} \parop \proj{\cntxt'[Q]}{\pid}{\emptyset} 
 \equiv \proj{P' \parop \cntxt'[Q]}{\pid}{\emptyset}$ for $\pid \neq \id$.
 
 (Case $\cntxt[\cdot]$ is $S \plusop \cntxt'[\cdot]$) 
 
 Follows similar lines.
 
(Case $\cntxt[\cdot]$ is $\idop{\id'} \cntxt'[\cdot]$)

We have that $P \equiv \idop{\id'} \cntxt'[\idop{\id}Q]$ and by definition (\emph{i})
$\proj{\idop{\id'}\cntxt'[\idop{\id}Q]}{\pid}{\emptyset}$ is $\proj{\cntxt'[\idop{\id}Q]}{\pid}{\emptyset} \parop \proj{\cntxt'[\idop{\id}Q]}{!}{\emptyset}$
in case $\pid = \id'$, and
$\proj{\idop{\id'}\cntxt'[\idop{\id}Q]}{\pid}{\emptyset}$ is $\proj{\cntxt'[\idop{\id}Q]}{\pid}{\emptyset}$ in case $\pid \neq \id'$,
regardless if $\id = \id'$ or $\id \neq \id'$.
In case $\pid \neq \id'$ the result follows directly from the induction hypothesis, where in particular we have 
(\emph{ii}) $\proj{\cntxt'[\idop{\id}Q]}{!}{\emptyset} \equiv \proj{\cntxt'[Q]}{!}{\emptyset}$.
For $\pid = \id'$ by induction hypothesis we have that 
(\emph{iii}) $\proj{\cntxt'[\idop{\id}Q]}{\id'}{\emptyset} \equiv \proj{\cntxt'[Q]}{\id'}{\emptyset} \parop \proj{Q}{!}{\emptyset}$ when $\id = \id'$ 
and that $\proj{\cntxt'[\idop{\id}Q]}{\id'}{\emptyset} \equiv \proj{\cntxt'[Q]}{\id'}{\emptyset}$ when $\id \neq \id'$. 
When $\id = \id'$ from (\emph{i}) and (\emph{iii}), and also by (\emph{ii}) and by definition of projection, 
we conclude $\proj{\idop{\id'} \cntxt'[\idop{\id}Q]}{\id'}{\emptyset} \equiv 
\proj{\cntxt'[Q]}{\id'}{\emptyset} \parop \proj{Q}{!}{\emptyset}
 \parop \proj{\cntxt'[\idop{\id}Q]}{!}{\emptyset} \equiv
 \proj{\cntxt'[Q]}{\id'}{\emptyset} \parop \proj{Q}{!}{\emptyset}
 \parop \proj{\cntxt'[Q]}{!}{\emptyset} 
 \equiv \proj{\idop{\id'}\cntxt'[Q]}{\id'}{\emptyset} \parop \proj{Q}{!}{\emptyset}$.
When $\id \neq \id'$ from (\emph{ii}) and by definition of projection we conclude $\proj{\idop{\id'} \cntxt'[\idop{\id}Q]}{\id'}{\emptyset} \equiv 
\proj{\cntxt'[Q]}{\id'}{\emptyset} 
 \parop \proj{\cntxt'[\idop{\id}Q]}{!}{\emptyset} \equiv
 \proj{\cntxt'[Q]}{\id'}{\emptyset} 
 \parop \proj{\cntxt'[Q]}{!}{\emptyset} 
 \equiv \proj{\idop{\id'}\cntxt'[Q]}{\id'}{\emptyset}$.

(Case $\cntxt[\cdot]$ is $\filter{o}{i}\cntxt'[\cdot]$) 

The case for $\id$ projection follows directly from
induction hypothesis since by definition $\proj{\filter{o}{i}\cntxt'[\idop{\id}Q]}{\id}{\sigma}$ is 
$\proj{\cntxt'[\idop{\id}Q]}{\id}{\sigma}$, and likewise for $\id'$ projection ($\id \neq \id'$). 
The case for $!$ projection is direct since by definition both
$\proj{\filter{o}{i}\cntxt'[\idop{\id}Q]}{!}{\sigma}$ 
and $\proj{\filter{o}{i}\cntxt'[Q]}{!}{\sigma}$ are
$\mconfig{o}\fun^{d}!$. The case for $?$ projection follows from the induction hypothesis in expected lines.
\end{proof}

The following result states that for each synchronisation action in the protocol there is a corresponding input 
in the respective $?$-projection, and that the continuation of the input is the reaction obtained by $!$-projecting the continuation 
of the synchronisation action.

\begin{lemma}[Synchronisation Action Projection]
\label{lem:filterproj}
If $P \equiv \cntxt[\filter{o}{i}Q]$ then $\proj{P}{?}{\emptyset} \equiv D \parop \mconfig{i}\fun^{d}?.{R}$
and $\proj{Q}{!}{\emptyset} \equiv R$.
\end{lemma}
\begin{proof}
By induction on the structure of $\cntxt[\cdot]$. Follows expected lines (cf. Lemma~\ref{lem:projactive}).
\end{proof}

The next result says $?$-projections can be described as a parallel composition of distinguished inputs.

\begin{lemma}[Input Uniqueness]
\label{lem:inpuni}
If $P$ is a well-formed protocol then $\proj{P}{?}{\sigma} \equiv \Pi_{i \in I} \mconfig{c_i}\fun_{i}^{d_i}?.{R_i}$ where $\fun_j \neq \fun_l$ for all $j,l \in I$ such that $j \neq l$.
\end{lemma}
\begin{proof}(Sketch)
Since well-formed protocol protocols originate from specifications where all action labels are distinct and we have that 
$?$ projection is invariant under reduction (Lemma~\ref{lem:projactive}) we may consider wlog that all action labels
are distinct in $P$. 
The proof then follows by induction on the structure of $P$ in expected lines (cf. Lemma~\ref{lem:projnormalform}(\emph{3})),
by preserving the invariant that the set of labels $\{ \fun_{i} \ | \ i \in I \}$ is the set of labels of $P$.
\end{proof}

The following result, auxiliary to Lemma~\ref{lem:projinv}, allows to describe protocols as a parallel composition of 
(active node scoped) synchronisation action summations.

\begin{lemma}[Protocol Normal Form]
\label{lem:protnorm}
Let $P$ be a protocol where recursion is guarded. We have that $P \equiv \idop{{\tilde{\ids_1}}} S_1 \parop \ldots \parop \idop{{\tilde{\ids_k}}} S_k$,
where $k \geq 0$.
\end{lemma}
\begin{proof}
By induction on the structure of $P$. 

\noindent{\bf (Case $P$ is $\filter{o}{i}Q$)}

Immediate.

\noindent{\bf(Case $P$ is $\idop{\id}Q$)} 

By induction hypothesis we have that $Q \equiv  \idop{{\tilde{\ids_1}}} S_1 \parop \ldots \parop \idop{{\tilde{\ids_k}}} S_k$,
hence $\idop{\id}Q \equiv \idop{\id}\idop{{\tilde{\ids_1}}} S_1 \parop \ldots \parop \idop{\id}\idop{{\tilde{\ids_k}}} S_k$.

\noindent{\bf(Case $P$ is $\zero$)}

Immediate.

\noindent{\bf(Case $P$ is $\varX$)}

Does not apply since recursion is guarded in $P$.

\noindent{\bf(Case $P$ is $\rec{Q}$)}

By induction hypothesis we have that $Q\equiv  \idop{{\tilde{\ids_1}}} S_1 \parop \ldots \parop \idop{{\tilde{\ids_k}}} S_k$,
hence we have that $\rec{Q}\equiv  \idop{{\tilde{\ids_1}}} (S_1[\rec{Q}/\varX]) \parop \ldots \parop \idop{{\tilde{\ids_k}}} (S_k[\rec{Q}/\varX])$.

\noindent{\bf(Case $P$ is $Q_1 \parop Q_2$)}

By induction hypothesis we have that $Q_1\equiv  \idop{{\tilde{\ids_1}}} S_1 \parop \ldots \parop \idop{{\tilde{\ids_k}}} S_k$
and $Q_2\equiv  \idop{{\tilde{\ids_1'}}} S_1' \parop \ldots \parop \idop{{\tilde{\ids_k'}}} S_k'$, hence
$Q_1 \parop Q_2 \equiv \idop{{\tilde{\ids_1}}} S_1 \parop \ldots \parop \idop{{\tilde{\ids_k}}} S_k \parop \idop{{\tilde{\ids_1'}}} S_1' \parop \ldots \parop \idop{{\tilde{\ids_k'}}} S_k'$.

\noindent{\bf(Case $P$ is $S_1 \plusop S_2$)}

Immediate.

\end{proof}

The following result relates choices yielded by an $\id$-projection to a summation in scope of a respective active node in the source protocol.

\begin{lemma}[Projection Inversion]
\label{lem:projinv}
Let $P$ be a well-formed protocol
such that $ \proj{P}{\id}{\emptyset} \equiv R \parop C$.
We have that $P \equiv \cntxt[\idop{\id}S]$ and $\proj{S}{!}{\emptyset} \equiv C$. 
\end{lemma}
\begin{proof}
By induction on the structure of $P$.

\noindent{\bf (Case $P$ is $\filter{o}{i}Q$)}

By definition we have that $\proj{\filter{o}{i}Q}{\id}{\emptyset}$ is $\proj{Q}{\id}{\emptyset}$.
By induction hypothesis we have that $Q \equiv \cntxt[\idop{\id}S]$ and $\proj{S}{!}{\emptyset} \equiv C$. Hence we have
$P \equiv \filter{o}{i}\cntxt[\idop{\id}S]$.

\noindent{\bf(Case $P$ is $\idop{\id'}Q$)} 

In case $\id \neq \id'$ the proof follows from the induction hypothesis. In case $\id = \id'$ we have that 
$\proj{\idop{\id}Q}{\id}{\emptyset} \equiv R \parop C$. By definition we have that $\proj{\idop{\id}Q}{\id}{\emptyset}$ 
is $\proj{Q}{\id}{\emptyset} \parop \proj{Q}{!}{\emptyset}$, hence $\proj{Q}{\id}{\emptyset} \parop \proj{Q}{!}{\emptyset} \equiv R \parop C$.
We either have that $\proj{Q}{\id}{\emptyset} \equiv R' \parop C$ and $\proj{Q}{!}{\emptyset} \equiv R''$ or
$\proj{Q}{\id}{\emptyset} \equiv R'$ and $\proj{Q}{!}{\emptyset} \equiv R'' \parop C$. In the former case the proof follows from the 
induction hypothesis. In the latter case, since $P$ is well-formed we have that recursion is guarded in $\idop{\id}Q$, from Lemma~\ref{lem:protnorm} 
we conclude $Q \equiv \idop{{\tilde{\ids_1}}} S_1 \parop \ldots \parop \idop{{\tilde{\ids_k}}} S_k$. We have that $k$ is greater than $0$ since
$C$ is a (non empty) choice. By definition we have
that $\proj{Q}{!}{\emptyset}$ is $\proj{S_1}{!}{\emptyset} \parop \ldots \parop \proj{S_k}{!}{\emptyset}$, hence
$\proj{S_1}{!}{\emptyset} \parop \ldots \parop \proj{S_k}{!}{\emptyset} \equiv R'' \parop C$. Thus there is $i \in 1,\ldots,k$ such that
$\proj{S_i}{!}{\emptyset} \equiv C$. Finally we conclude $\idop{\id}Q \equiv 
\idop{\id}\idop{{\tilde{\ids_1}}} S_1 \parop \ldots \parop \idop{{\tilde{\ids_i}}}\idop{\id} S_i\parop\ldots \parop \idop{\id}\idop{{\tilde{\ids_k}}} S_k$.

\noindent{\bf(Case $P$ is $\zero$)}

Impossible since $C$ is a (non empty) choice.

\noindent{\bf(Case $P$ is $\varX$)}

Impossible since $C$ is a (non empty) choice ($ \proj{\varX}{\id}{\emptyset}$ is $\zero$).

\noindent{\bf(Case $P$ is $\rec{Q}$)}

Impossible since $P$ is a well-formed protocol ($ \proj{\rec{Q}}{\id}{\emptyset}$ is $\zero$).

\noindent{\bf(Case $P$ is $Q_1 \parop Q_2$)}

Follows from the induction hypothesis.

\noindent{\bf(Case $P$ is $S_1 \plusop S_2$)}

Follows from the induction hypothesis.
\end{proof}

The following results characterise the structure of networks and definitions given a specific observable behaviour.

\begin{lemma}[Inversion on Network Internal Step] 
\label{lem:invnettau}
If $\net\rTo{\tau}{}{\net'}$ then either 
 \begin{itemize}
\item $\net\equiv\net_1 \nparop \net_2$ and $\net_1\rTo{\linktr{\id_1\rightarrow\id_2}{\fun}}{}\net_1'$, 
$\net_2\rTo{\linktr{\id_2\leftarrow\id_1}{\fun}}{}\net_2'$, and $\net' \equiv \net_1' \nparop \net_2'$.
\item $\net\equiv\node{\state}{D}\nparop \net''$ or $\net\equiv\node{\state}{D} $, $\node{\state}{D}\rTo{\tau}\node{\state}{D'}$
and $\net' \equiv \node{\state}{D'}\nparop \net''$ or $\net' \equiv \node{\state}{D'}$ respectively.
\end{itemize}
\end{lemma}
\begin{proof}
By induction on the derivation of $\net\rTo{\tau}{}{\net'}$ following expected lines.
\end{proof}

\begin{lemma}[Inversion on Definition Output]
\label{lem:invdefout}
 If $D\rTo{\mconfig{c}\fun^{d}!} D'$ then $D \equiv D' \parop \mconfig{c}\func{d}{}!+C$.
\end{lemma}
\begin{proof}
By induction on the derivation of $D\rTo{\mconfig{c}\fun^{d}!} D'$ following expected lines.
\end{proof}

\begin{lemma}[Inversion on Network Output]
\label{lem:invnetout}
 If $\net\rTo{\linktr{\id_1\rightarrow\id_2}{\fun}}{}{\net'}$ then $\net\equiv\node{\state}{D \parop \mconfig{o}\func{d}{}! \plusop C} \nparop \net''$ or
$\net\equiv\node{\state}{D \parop \mconfig{o}\func{d}{}!\plusop C}$, $\net'\equiv\node{\state'}{D} \nparop \net''$ or $\net'\equiv\node{\state'}{D}$ respectively, $\state'=\fun^d!(\state,\id_2)$, $\state\models o$ and $\id(\state) =\id$. Furthermore if $d =\ \toright$ then $\id_2 \in \nei(\state)$ and if 
$d = \toup$ then
$\id_2 = i(\state)$.
\end{lemma}
\begin{proof}
By induction on the derivation of $\net\rTo{\linktr{\id_1\rightarrow\id_2}{\fun}}{}{\net'}$, where the base case 
follows by inversion on $\rulename{oBinR}$ and $\rulename{oBinU}$ and by considering Lemma~\ref{lem:invdefout}.
\end{proof}

\begin{lemma}[Inversion on Definition Input]
\label{lem:invdefinp}
 If $D\rTo{\mconfig{c}\fun^{d}?} D'$ then $D \equiv D'' \parop \mconfig{c}\func{d}{}?.R$ and $D' \equiv D'' \parop \mconfig{c}\func{d}{}?.R \parop R$.
\end{lemma}
\begin{proof}
By induction on the derivation of $D\rTo{\mconfig{c}\fun^{d}?} D'$ following expected lines.
\end{proof}

\begin{lemma}[Inversion on Network Input]
\label{lem:invnetinp}
 If $\net\rTo{\linktr{\id_1\leftarrow\id_2}{\fun}}{}{\net'}$ then $\net\equiv\node{\state}{D \parop \mconfig{i}\func{d}{}?.R}\nparop \net''$ 
 or $\net\equiv\node{\state}{D \parop \mconfig{i}\func{d}{}?.R}$, $\net'\equiv\node{\state'}{D \parop \mconfig{i}\func{d}{}?.R \parop R}\nparop \net''$  
 or $\net'\equiv\node{\state'}{D \parop \mconfig{i}\func{d}{}?.R \parop R}$ respectively, $\state'=\fun^d?(\state,\id_2)$, $\state\models i$,
 and $\id(\state) = \id_1$.
\end{lemma}
\begin{proof}
By induction on the derivation of $\net\rTo{\linktr{\id_1\rightarrow\id_2}{\fun}}{}{\net'}$, where the base case 
follows by inversion on $\rulename{iBin}$ and by considering Lemma~\ref{lem:invdefinp}.
\end{proof}

\begin{lemma}[Inversion on Definition Internal Step]
\label{lem:invdeftau}
 If $D\rTo{\mconfig{c}\fun} D'$ then 
$D \equiv \mconfig{i}\func{\toself}{}?.R \parop \mconfig{o}\func{\toself}{}!+C \parop D''$ and $D' \equiv \mconfig{i}\func{\toself}{}?.R \parop R \parop D''$ 
and $c = i \wedge o$.
\end{lemma}
\begin{proof}
By induction on the derivation of $D\rTo{\mconfig{c}\fun} D'$ where the base case follows by inversion on $\rulename{Self}$ and
by considering Lemma~\ref{lem:invdefout} and Lemma~\ref{lem:invdefinp}.
\end{proof}

\begin{lemma}[Inversion on Network Broadcast Input]
\label{lem:invbroinp}
 If $\net\rTo{\linktr{\id?\star}{\fun}}{}{\net'}$ then $\net\equiv\Pi_{ j\in J}(\node{\state_j}{D_j})$ and 
$\net'\equiv\Pi_{ j\in J}(\node{\state_j}{D_j'})$ and $\forall j\in J$ it is the case
 that either ({i}) $D_j \equiv D_j'' \parop \mconfig{i}\func{\star}{}?.R_j$ and $s_j \models i$ and $i(\state_j) = \id$ and 
 $D_j' \equiv D_j'' \parop \mconfig{i}\func{\star}{}?.R_j \parop R_j$ 
 or ({ii}) $\mathit{discard}(\node{\state_j}{D_j},\linktr{\id?\toall}{\fun})$ and $D_j' \equiv D_j$.
\end{lemma}
\begin{proof}
By induction on the derivation of $\net\rTo{\linktr{\id?\star}{\fun}}{}{\net'}$, where the base case follows by inversion on $\rulename{iBrd}$
and by considering Lemma~\ref{lem:invdefinp} or by inversion on $\rulename{dBrd}$.
\end{proof}

\begin{lemma}[Inversion on Network Broadcast Output]
\label{lem:invbroout}
If $\net\rTo{\linktr{\id!\star}{\fun}}{}{\net'}$ then $\net\equiv\node{\state}{D\ \parop\ \mconfig{o}\func{\star}{}!+C}\ \|\ \net_1$ 
or $\net\equiv\node{\state}{D\ \parop\ \mconfig{o}\func{\star}{}!+C}$, $\net' \equiv \node{\state}{D}\ \|\ \net_2$ and 
$\net_1\rTo{\linktr{\id?\star}{\fun}}{}\net_2$ or $\net' \equiv \node{\state}{D}$, respectively, $\state\models o$ and $\id(\state) =\id$.
\end{lemma}
\begin{proof}
By induction on the derivation of $\net\rTo{\linktr{\id!\star}{\fun}}{}{\net'}$, where the base case follows by inversion on $\rulename{oBrd}$
and by considering Lemma~\ref{lem:invdefout}.
\end{proof}

The following result attests reduction is closed under active contexts.

\begin{lemma}[Reduction Closed Under Active Contexts]
\label{lem:redclosurecntxt}
If $\sys{\Delta}{Q}  \rTo{}{} \sys{\Delta'}{Q'}$ then $\sys{\Delta}{\cntxt[Q]} \rTo{}{} \sys{\Delta'}{\cntxt[Q']}$. 
\end{lemma}
\begin{proof}
By induction on the structure of $\cntxt[\cdot]$ following expected lines as for each possible active context there is a corresponding reduction rule.
\end{proof}

We may now state the completeness of our projection relation, i.e., that behaviours exhibited by the yielded projection have a correspondent
step in the source configuration. The structure of the proof is as follows: first we use the inversion results that allow us to characterise the
structure of the definitions yielded by the projection, namely to identify the interacting input/output definitions. Then, we relate the output (choice)
to a synchronisation action in the protocol (Projection Inversion, Lemma~\ref{lem:projinv}). Since the identified synchronisation action in the 
protocol must have an unique input counterpart in the projection (Synchronisation Action Projection, Lemma~\ref{lem:filterproj}, and Input 
Uniqueness, Lemma~\ref{lem:inpuni}) we obtain the relation of the synchronisation action with the input involved in the interaction. Having 
identified the related synchronisation action we are able to derive the reduction, at which point we proceed to prove the relation between
the arrival configuration with the network resulting from the interaction step. Crucial to this part of the proof is the characterisation of the
Active Node Projection (Lemma~\ref{lem:projactive}) since reduction in configurations, in essence, is a change of active nodes
(cf. the reduction semantics): first the enabling node is active then the reacting nodes are activated. By means of Lemma~\ref{lem:projactive}
we are able to pinpoint that the projection of the enabling node differs only in the presence of the output (choice) before interaction, absent 
after interaction. We are also able to pinpoint that the projection of the reacting node differs only in the activation of the reaction in the
continuation of the input after interaction, not active before interaction. Lemma~\ref{lem:projactive} also allows to conclude that nodes
not involved in the interaction have identical projections before and after the interaction. The precise characterisation allows us to relate 
the network resulting from the interaction step to the projection of the arrival configuration.

\begin{lemma}[Completeness]
\label{lem:completeness}
If $\tproj{\sys{\Delta}{P}} \rTo{\hat{\lambda}}{} \net$ then
$\sys{\Delta}{P} \rTo{}{} \sys{\Delta'}{Q}$ and $\net \equiv \tproj{\sys{\Delta'}{Q}}$.
\end{lemma}
\begin{proof}(Sketch)
We detail the proof for the case $\hat{\lambda} = \tau$ (in particular for child to parent binary interaction), remaining cases follow similar lines
(in particular for broadcast interaction which proof builds on Lemma~\ref{lem:invbroinp} and Lemma~\ref{lem:invbroout}).

By Lemma~\ref{lem:invnettau} we have that (1)
$\tproj{\sys{\Delta}{P}}\equiv\net_1 \nparop \net_2$ and $\net_1\rTo{\linktr{\id_1\rightarrow\id_2}{\fun}}{}\net_1'$,
$\net_2\rTo{\linktr{\id_2\leftarrow\id_1}{\fun}}{}\net_2'$, and $\net \equiv \net_1' \nparop \net_2'$ or 
(2) $\tproj{\sys{\Delta}{P}}\equiv\node{\state}{D}\nparop \net''$ or $\net\equiv\node{\state}{D} $, 
$\node{\state}{D}\rTo{\tau}\node{\state}{D'}$ and $\net \equiv \node{\state}{D'}\nparop \net''$ or $\net \equiv \node{\state}{D'}$ respectively. 
We consider (1) and remark that (2) follows similar lines (inversion on $\rulename{Loc}$ 
and the use of Lemma~\ref{lem:invdeftau} are the main differences).

From $\net_1\rTo{\linktr{\id_1\rightarrow\id_2}{\fun}}{}\net_1'$ and considering Lemma~\ref{lem:invnetout} we have that
$\net_1\equiv\node{\state_1}{D_1 \parop \mconfig{o}\func{d}{}! \plusop C} \nparop \net_1''$ or
$\net_1\equiv\node{\state_1}{D_1 \parop \mconfig{o}\func{d}{}!\plusop C}$, $\net_1'\equiv\node{\state_1'}{D_1} \nparop \net_1''$ or 
$\net_1'\equiv\node{\state_1'}{D_1}$ correspondingly, $\state_1'=\fun^d!(\state_1,\id_2)$, $\state_1 \models o$ and $\id(\state_1) =\id_1$. 
Furthermore if $d =\ \toright$ then $\id_2 \in \nei(\state_1)$ and if $d = \toup$ then
$\id_2 = i(\state_1)$. We consider the case when $\net_1\equiv\node{\state_1}{D_1 \parop \mconfig{o}\func{d}{}! \plusop C} \nparop \net_1''$
and $\net_1'\equiv\node{\state_1'}{D_1} \nparop \net_1''$ and when $d = \toup$ hence $\id_2 = i(\state_1)$.

By definition of projection, we have that $\proj{P}{?}{\emptyset} \parop \proj{P}{\id_1}{\emptyset} \equiv D_1 \parop \mconfig{o}\func{\toup}{}! \plusop C$.
Considering Lemma~\ref{lem:projnormalform} we have that there are $D_1'$ and $R_1$ such that $D_1 \equiv D_1' \parop R_1$ and 
$\proj{P}{?}{\emptyset} \equiv D_1'$ and $\proj{P}{\id_1}{\emptyset} \equiv R_1 \parop \mconfig{o}\func{\toup}{}! \plusop C$.
From Lemma~\ref{lem:projinv} we conclude
(\emph{i}) $P \equiv \cntxt[\idop{\id_1}S]$ and $\proj{S}{!}{\emptyset} \equiv \mconfig{o}\func{\toup}{}! \plusop C$,
 hence $S \equiv \ufilter{o}{i}P' \plusop S'$.

From $\net_2\rTo{\linktr{\id_2\leftarrow\id_1}{\fun}}{}\net_2'$ and considering Lemma~\ref{lem:invnetinp} we have that
$\net_2\equiv\node{\state_2}{D_2 \parop \mconfig{i'}\func{d'}{}?.R}\nparop \net_2''$ 
 or $\net_2\equiv\node{\state_2}{D_2 \parop \mconfig{i'}\func{d'}{}?.R}$, and that 
 $\net_2'\equiv\node{\state_2'}{D_2 \parop \mconfig{i'}\func{d'}{}?.R \parop R}\nparop \net_2''$  
 or $\net_2'\equiv\node{\state_2'}{D_2 \parop \mconfig{i'}\func{d'}{}?.R \parop R}$ respectively, and that
 $\state_2'=\fun^{d'}?(\state_2,\id_1)$, $\state_2\models i'$,
 and $\id(\state_2) = \id_2$. We consider the case when 
 $\net_2\equiv\node{\state_2}{D_2 \parop \mconfig{i'}\func{d'}{}?.R}\nparop \net_2''$ 
 and $\net_2'\equiv\node{\state_2'}{D_2 \parop \mconfig{i'}\func{d'}{}?.R \parop R}\nparop \net_2''$.
 
 By definition of projection, we have that $\proj{P}{?}{\emptyset} \parop \proj{P}{\id_2}{\emptyset} \equiv D_2 \parop \mconfig{i'}\func{d'}{}?.R$.
 Considering Lemma~\ref{lem:projnormalform} we have that there are $D_2'$ and $R_2$ such that $D_2 \equiv D_2' \parop R_2$ and 
 $\proj{P}{?}{\emptyset} \equiv D_2' \parop \mconfig{i'}\func{d'}{}?.R$ and $\proj{P}{\id_2}{\emptyset} \equiv R_2$.
 From (\emph{i}) and considering Lemma~\ref{lem:filterproj} we conclude 
$\proj{P}{?}{\emptyset} \equiv D' \parop \mconfig{i}\fun^{\toup}?.{R'}$ and $\proj{P'}{!}{\emptyset} \equiv R'$.
From Lemma~\ref{lem:inpuni} we conclude $i = i'$, $d' = \toup$ and $R \equiv R'$, hence  $\proj{P'}{!}{\emptyset} \equiv R$.

We then have that $\state_1 \models o$, $\id_2 = i(\state_1)$, $\state_2\models i$ and we may obtain $\Delta'$ by considering
$\state_1'=\fun^\toup!(\state_1,\id_2)$ and  $\state_2'=\fun^{\toup}?(\state_2,\id_1)$. Considering Lemma~\ref{lem:redclosurecntxt}
and by $\rulename{Bin}$ we conclude
$\sys{\Delta}{\cntxt[\idop{\id_1}\ufilter{o}{i}P' \plusop S']}  \rTo{}{} \sys{\Delta'}{\cntxt[\ufilter{o}{i}(\idop{\id_2}P') \plusop S']}$.

From the fact that $\net \equiv \net_1' \nparop \net_2'$ we conclude (\emph{ii})
$$\net \equiv \node{\state_1'}{D_1} \nparop \net_1'' \nparop
\node{\state_2'}{D_2 \parop \mconfig{i}\func{\toup}{}?.R \parop R}\nparop \net_2''$$

We recall that $\proj{P}{?}{\emptyset} \parop \proj{P}{\id_1}{\emptyset} \equiv D_1 \parop \mconfig{o}\func{\toup}{}! \plusop C$
and $D_1 \equiv D_1' \parop R_1$ and $\proj{P}{?}{\emptyset} \equiv D_1'$ and $\proj{P}{\id_1}{\emptyset} \equiv R_1 \parop \mconfig{o}\func{\toup}{}! \plusop C$ and that $P \equiv \cntxt[\idop{\id_1}\ufilter{o}{i}P' \plusop S']$. By Lemma~\ref{lem:projactive} we conclude that
$\proj{P}{\id_1}{\emptyset} \equiv \proj{\cntxt[\ufilter{o}{i}P' \plusop S']}{\id_1}{\emptyset} \parop \proj{\ufilter{o}{i}P' \plusop S'}{!}{\emptyset}$
and also that $\proj{P}{?}{\emptyset} \equiv \proj{\cntxt[\ufilter{o}{i}P' \plusop S']}{?}{\emptyset}$. Again by Lemma~\ref{lem:projactive}
and considering $\id_2 \neq id_1$ (and $? \neq \id_2$) we conclude $\proj{\cntxt[\ufilter{o}{i}P' \plusop S']}{\id_1}{\emptyset} 
\equiv \proj{\cntxt[\ufilter{o}{i}(\idop{\id_2}P') \plusop S']}{\id_1}{\emptyset}$
and $\proj{\cntxt[\ufilter{o}{i}P' \plusop S']}{?}{\emptyset} \equiv \proj{\cntxt[\ufilter{o}{i}(\idop{\id_2}P') \plusop S']}{?}{\emptyset}$.
From the latter and considering $\proj{P}{?}{\emptyset} \equiv D_1'$ we conclude $\proj{\cntxt[\ufilter{o}{i}(\idop{\id_2}P') \plusop S']}{?}{\emptyset}
\equiv D_1'$. From $\proj{P}{\id_1}{\emptyset} \equiv R_1 \parop \mconfig{o}\func{\toup}{}! \plusop C$ and
$\proj{P}{\id_1}{\emptyset} \equiv \proj{\cntxt[\ufilter{o}{i}P' \plusop S']}{\id_1}{\emptyset} \parop \proj{\ufilter{o}{i}P' \plusop S'}{!}{\emptyset}$
and $\proj{ \ufilter{o}{i}P' \plusop S'}{!}{\emptyset} \equiv \mconfig{o}\func{\toup}{}! \plusop C$
we conclude $\proj{\cntxt[\ufilter{o}{i}P' \plusop S']}{\id_1}{\emptyset} \equiv R_1$.
Then, by Lemma~\ref{lem:projactive} we conclude  $\proj{\cntxt[\ufilter{o}{i}(\idop{\id_2}P') \plusop S']}{\id_1}{\emptyset} \equiv R_1$
since $\id_2 \neq \id_1$.
We then have that $\proj{\cntxt[\ufilter{o}{i}(\idop{\id_2}P') \plusop S']}{?}{\emptyset} \equiv D_1'$
and $\proj{\cntxt[\ufilter{o}{i}(\idop{\id_2}P') \plusop S']}{\id_1}{\emptyset} \equiv R_1$ and hence, considering $D_1 \equiv D_1' \parop R_1$,
we have that (\emph{iii})
$$\node{\state_1'}{D_1} \equiv \node{\state_1'}{\proj{\cntxt[\ufilter{o}{i}(\idop{\id_2}P') \plusop S']}{?}{\emptyset} \parop \proj{\cntxt[\ufilter{o}{i}(\idop{\id_2}P') \plusop S']}{\id_1}{\emptyset}}$$

We recall that $\proj{P}{?}{\emptyset} \parop \proj{P}{\id_2}{\emptyset} \equiv D_2 \parop \mconfig{i}\func{\toup}{}?.R$ and
$D_2 \equiv D_2' \parop R_2$ and $\proj{P}{?}{\emptyset} \equiv D_2' \parop \mconfig{i}\fun^{\toup}?.{R}$ and 
 $\proj{P}{\id_2}{\emptyset} \equiv R_2$ and that $P \equiv \cntxt[\idop{\id_1}\ufilter{o}{i}P' \plusop S']$.
By Lemma~\ref{lem:projactive}, since $? \neq \id_1$ and $? \neq \id_2$, we conclude 
 $\proj{\cntxt[\idop{\id_1}\ufilter{o}{i}P' \plusop S']}{?}{\emptyset} \equiv \proj{\cntxt[\ufilter{o}{i}P' \plusop S']}{?}{\emptyset}
 \equiv \proj{\cntxt[\ufilter{o}{i}(\idop{\id_2}P') \plusop S']}{?}{\emptyset}$ and hence we have that
 $\proj{\cntxt[\ufilter{o}{i}(\idop{\id_2}P') \plusop S']}{?}{\emptyset} \equiv D_2' \parop \mconfig{i}\fun^{\toup}?.{R}$.
From Lemma~\ref{lem:projactive} we conclude 
$\proj{\cntxt[\ufilter{o}{i}(\idop{\id_2}P') \plusop S']}{\id_2}{\emptyset} \equiv \proj{\cntxt[\ufilter{o}{i}P' \plusop S']}{\id_2}{\emptyset}
\parop \proj{P'}{!}{\emptyset}$. Since $\proj{P'}{!}{\emptyset} \equiv R$ we conclude 
$\proj{\cntxt[\ufilter{o}{i}(\idop{\id_2}P') \plusop S']}{\id_2}{\emptyset} \equiv \proj{\cntxt[\ufilter{o}{i}P' \plusop S']}{\id_2}{\emptyset}
\parop R$ and since $\id_1 \neq \id_2$, again considering Lemma~\ref{lem:projactive}, we conclude
$\proj{\cntxt[\ufilter{o}{i}(\idop{\id_2}P') \plusop S']}{\id_2}{\emptyset} \equiv \proj{\cntxt[\idop{\id_1}\ufilter{o}{i}P' \plusop S']}{\id_2}{\emptyset}
\parop R$. Since  $\proj{P}{\id_2}{\emptyset} \equiv R_2$ we conclude 
$\proj{\cntxt[\ufilter{o}{i}(\idop{\id_2}P') \plusop S']}{\id_2}{\emptyset} \equiv R_2 \parop R$.
We then have that  $\proj{\cntxt[\ufilter{o}{i}(\idop{\id_2}P') \plusop S']}{?}{\emptyset} \equiv D_2' \parop \mconfig{i}\fun^{\toup}?.{R}$
and $\proj{\cntxt[\ufilter{o}{i}(\idop{\id_2}P') \plusop S']}{\id_2}{\emptyset} \equiv R_2 \parop R$,
and hence, considering $D_2 \equiv D_2' \parop R_2$ we conclude (\emph{iv})
$$
\node{\state_2'}{D_2 \parop \mconfig{i}\func{\toup}{}?.R \parop R} \equiv
\node{\state_2'}{\proj{\cntxt[\ufilter{o}{i}(\idop{\id_2}P') \plusop S']}{?}{\emptyset} \parop 
\proj{\cntxt[\ufilter{o}{i}(\idop{\id_2}P') \plusop S']}{\id_2}{\emptyset}}
$$

Considering Lemma~\ref{lem:projactive}, for any $\id'$ such that $\id' \neq \id_1$ and $\id'\neq\id_2$, we conclude
$\proj{P}{?}{\emptyset} \parop \proj{P}{\id'}{\emptyset} \equiv
\proj{\cntxt[\ufilter{o}{i}(\idop{\id_2}P') \plusop S']}{?}{\emptyset} \parop 
\proj{\cntxt[\ufilter{o}{i}(\idop{\id_2}P') \plusop S']}{\id'}{\emptyset}$. Since also the changes in $\Delta'$ are localised to $\state_1'$ and $\state_2'$
we may show that $\net_1''$ and $\net_2''$ may be obtained correspondingly.
Considering (\emph{ii}), (\emph{iii}) and (\emph{iv}) we may then conclude
$\net \equiv \tproj{\sys{\Delta'}{\cntxt[\ufilter{o}{i}(\idop{\id_2}P') \plusop S']}}$.

\end{proof}

\bibliography{biblio}
\end{document}